\documentclass[final]{siamltex}
\usepackage{epsfig}
\usepackage{color}
\usepackage{amsmath}
\usepackage{amscd}
\usepackage{amsmath, amsfonts, amssymb,mathrsfs}
\usepackage{multirow}
\usepackage{enumerate}
\usepackage{verbatim}
\usepackage{cite}
\usepackage{tikz}
\usetikzlibrary{matrix,arrows}
\usepackage[justification=raggedright]{caption}
\usepackage{subcaption}
\usepackage[ruled]{algorithm2e}

\numberwithin{equation}{section}

\def \vec#1{{\bf{#1}}}
\def\pd#1#2{\frac{\partial #1}{\partial #2}}

\newcommand{\bi}{\begin{itemize}}
\newcommand{\ei}{\end{itemize}}

\newcommand{\diverg}{\vec{\nabla}\cdot}
\newcommand{\director}{\vec{n}}
\newcommand{\curl}{\vec{\nabla}\times}
\newcommand{\Ltwoinner}[3]{\langle #1,#2 \rangle_0}
\newcommand{\Ltwonorm}[2]{\Vert #1 \Vert_0}
\newcommand{\Ltwonormndim}[3]{\Vert #1 \Vert_0}
\newcommand{\Ltwoinnerndim}[4]{\langle #1,#2 \rangle_0}

\newcommand{\diff}[1]{\, d#1}

\newcommand{\ltwonorm}[1]{\vert #1 \vert}
\newcommand{\ltwoinner}[2]{( #1, #2 )}
\newcommand{\kdirector}{\vec{n}_k}
\newcommand{\ddirector}{\delta \director}
\newcommand{\dlambda}{\delta \lambda}
\newcommand{\klambda}{\lambda_k}
\DeclareMathOperator*{\argmin}{argmin}

\newcommand{\kphi}{\phi_k}

\newcommand{\dphi}{\delta \phi}
\newcommand{\lagdivn}{\mathcal{L}_{\director}[\vec{v}]}
\newcommand{\lagdivlam}{\mathcal{L}_{\lambda}[\gamma]}

\newcommand{\lagdivphi}{\mathcal{L}_{\phi}[\psi]}

\newcommand{\Honenorm}[2]{\Vert #1 \Vert_1}

\newcommand{\Hdc}{\mathcal{H}^{DC}{(\Omega)}}
\newcommand{\Hdcnot}{\mathcal{H}^{DC}_0{(\Omega)}}

\newcommand{\Hone}[1]{H^1(#1)}

\newcommand{\Honenot}[1]{H^1_0({#1})}
\newcommand{\Honebnot}[1]{H^{1,0}({#1})}

\newcommand{\Honeb}[1]{H^{1,g}(#1)}

\newcommand{\Ltwo}[1]{L^2(#1)}

\newcommand{\Lp}[1]{L^p (\Omega)}
\newcommand{\Hcurl}[1]{H(\text{curl},#1)}
\newcommand{\Hdiv}[1]{H(\text{div},#1)}

\newcommand{\Linfinity}[1]{L^{\infty}(\Omega)}

\newcommand{\Hdcnorm}[2]{\Vert #1 \Vert_{DC}}

\newcommand{\triangulation}{\mathcal{T}_h}
\newcommand{\Pflex}{\vec{P}_{\text{flexo}}}

\newcounter{casenum}
\newenvironment{caseof}{\setcounter{casenum}{1}}{}
\newcommand{\case}[2]{\vskip.5\baselineskip\par\noindent {\bfseries Case \arabic{casenum}.} #1 \\ #2\addtocounter{casenum}{1}}

\title{Energy Minimization for Liquid Crystal Equilibrium with Electric and Flexoelectric Effects\thanks{Submitted June 28, 2014}}

\author{J. H. Adler, T. J. Atherton, T. R. Benson, D. B. Emerson, S. P. MacLachlan}

\begin{document}

\maketitle

\begin{abstract}
This paper outlines an energy-minimization finite-element approach to the modeling of equilibrium configurations for nematic liquid crystals in the presence of internal and external electric fields. The method targets minimization of system free energy based on the electrically and flexoelectrically augmented Frank-Oseen free energy models. The Hessian, resulting from the linearization of the first-order optimality conditions, is shown to be invertible for both models when discretized by a mixed finite-element method under certain assumptions. This implies that the intermediate discrete linearizations are well-posed. A coupled multigrid solver with Vanka-type relaxation is proposed and numerically vetted for approximation of the solution to the linear systems arising in the linearizations. Two electric model numerical experiments are performed with the proposed iterative solver. The first compares the algorithm's solution of a classical Freedericksz transition problem to the known analytical solution and demonstrates the convergence of the algorithm to the true solution. The second experiment targets a problem with more complicated boundary conditions, simulating a nano-patterned surface. In addition, numerical simulations incorporating these nano-patterned boundaries for a flexoelectric model are run with the iterative solver. These simulations verify expected physical behavior predicted by a perturbation model. The algorithm accurately handles heterogeneous coefficients and efficiently resolves configurations resulting from classical and complicated boundary conditions relevant in ongoing research.
 \end{abstract}
 

 
 \pagestyle{myheadings}
\thispagestyle{plain}
\markboth{\sc Adler, Atherton, Benson, Emerson, Maclachlan}{\sc Nematic Liquid Crystal Energy Minimization}

\section{Introduction}

Liquid crystals, whose discovery is attributed to Reinitzer in 1888 \cite{Reinitzer1}, are substances that possess mesophases with properties intermediate between liquids and crystals, existing at different temperatures or solvent concentrations. The focus of this paper is on nematic liquid crystal phases, which are formed by rod-like molecules that self-assemble into an ordered structure, such that the molecules tend to align along a preferred orientation. The preferred average direction at any point in a domain $\Omega$ is known as the director, denoted $\director(x,y,z) = (n_1, n_2, n_3)^T$. The director is taken to be of unit length at every point and headless, that is $\director$ and $-\director$ are indistinguishable, reflecting the observed symmetry of the phase.

In addition to their self-structuring properties, nematic liquid crystals are dielectrically active. Thus, their configurations are affected by electric fields. Moreover, since these materials are birefringent, with refractive indices that depend on the polarization of light, they can be used to control the propagation of light through a nematic structure. These traits have led, and continue to lead, to important discoveries in display technologies and beyond \cite{Lagerwall1}. Modern applications include  nanoparticle organization, liquid crystal-functionalized polymer fibers \cite{Lagerwall1}, and liquid crystal elastomers designed to produce effective actuator devices such as light driven motors \cite{Yamada1} and artificial muscles \cite{Thomsen1}. Thorough overviews of liquid crystal physics are found in \cite{Stewart1, deGennes1, Chandrasekhar1}. 

Many mathematical and computational models of liquid crystal continuum theory lead to complicated systems involving unit length constrained vector fields. Currently, the complexity of such systems has restricted the existence of known analytical solutions to simplified geometries in one (1-D) or two dimensions (2-D), often under strong simplifying assumptions. When coupled with electric fields and other effects, far fewer analytical solutions exist, even in 1-D \cite{Stewart1}. In addition, associated systems of partial differential equations, such as the equilibrium equations \cite{Stewart1, Ericksen4}, suffer from non-unique solutions, which must be distinguished via energy arguments. Due to such difficulties, efficient, theoretically supported, numerical approaches to the modeling of nematic liquid crystals under free elastic and augmented electric effects are of great importance. A number of computational techniques for liquid crystal equilibrium and dynamics problems exist \cite{Liu1, Liu2, Liu3, Stewart1}, including least-squares finite-element methods \cite{Atherton1} and discrete Lagrange multiplier approaches with simplifying assumptions \cite{Ramage1, Ramage2}. In addition, numerical experiments involving finite-element methods with Lagrange multipliers, applied to the equilibrium equations, have been successful in capturing certain liquid crystal characteristics \cite{Pandolfi1}. 

In this paper, we propose a method that directly targets energy minimization in the continuum, via Lagrange multiplier theory on Banach spaces, to resolve liquid crystal equilibrium configurations in the presence of applied electric fields and internally induced electric fields due to flexoelectric effects. The approach is derived absent the often used one-constant approximation \cite{Ramage1, Liu1, Stewart1, Cohen1}; that is, the method described here, and the accompanying theory, are applicable for a wide range of physical parameters. This allows for significantly improved modeling of physical phenomena not captured in many models. Furthermore, most models and analytical approaches rely on assumptions to reduce the dimensionality of the problem. Here, the method and theory are suitable for use on 2-D and three dimensional (3-D) domains.

After defining the energy functional to be minimized, first-order optimality conditions are computed. These first-order conditions contain highly nonlinear terms and are, therefore, linearized with a generalized Newton's method. The resulting Newton iteration inherently contains a complicated saddle-point structure \cite{Benzi2, Ramage1}. The discrete Hessians associated with finite-element discretization of the Newton linearizations are shown to be invertible, for both the electric and flexoelectric models, when employing certain finite-element spaces. 

In addition, we discuss a coupled multigrid solver with Vanka-type relaxation for accurate and efficient resolution of solutions to the saddle-point systems encountered in the discretization of the linearization systems for both the electric and flexoelectric models. A full, mesh-cell oriented Vanka-type relaxation technique is elaborated and implemented. The performance of the multigrid solver is compared to that of a direct LU decomposition approach. Furthermore, it is applied to a collection of numerical examples, demonstrating its accuracy and efficiency.

This paper is organized as follows. We first introduce the electric field model under consideration in Section \ref{energymodels}. The method framework is derived and Dirichlet boundary condition simplifications are discussed in Section \ref{freeenergymin}. In Section \ref{wellposedhessian}, the invertibility of the discretized Hessian for the intermediate Newton linearizations is established. An extension of the method and associated theory for the flexoelectric model is given in Section \ref{flexoaugmentation}. The numerical methodology, iterative solver, and numerical experiments are detailed in Sections \ref{nummethodology} and \ref{numresults}. Finally, Section \ref{conclusion} gives some concluding remarks, and future work is discussed.

\section{Energy Model} \label{energymodels}

To begin defining the full energy model under consideration, we first discuss the free elastic energy model. At equilibrium, absent any external forces, fields, or boundary conditions, the free elastic energy present in a liquid crystal sample is given by an integral functional that depends on the state variables of the system. A liquid crystal sample tends to the state of lowest free energy. While a number of free energy models exist \cite{Davis1}, this paper considers the Frank-Oseen free elastic model \cite{Stewart1, Virga1}. This model represents the free elastic energy density, $w_F$, in a sample as
\begin{align} \label{FrankOseenFree}
w_F &= \frac{1}{2}K_1(\diverg \director)^2+ \frac{1}{2}K_2(\director \cdot \curl \director)^2+ \frac{1}{2}K_3\vert \director \times \curl \director \vert^2 \nonumber \\
& \qquad+ \frac{1}{2}(K_2+K_4)\diverg[(\director \cdot \vec{\nabla}) \director - (\diverg \director) \director].
\end{align}
Throughout this paper, the standard Euclidean inner product and norm are denoted $(\cdot, \cdot)$ and $\vert \cdot \vert$, respectively. The $K_i$, $i=1,2,3,4$, are known as the Frank elastic constants \cite{Frank1}, which vary depending on temperature and liquid crystal type.  As in \cite{Emerson1}, let
\begin{equation*} \label{matrixD}
\vec{Z} = \kappa \director \otimes \director + (\vec{I} - \director \otimes \director) = \vec{I} - (1-\kappa) \director \otimes \director,
\end{equation*}
where $\kappa = K_2/K_3$ with $K_2, K_3 \geq 0$ by Ericksen's inequalities \cite{Ericksen2}. In general, we consider the case that  $K_2, K_3 \neq 0$. Denote the classical $\Ltwo{\Omega}$ inner product and norm as $\Ltwoinner{\cdot}{\cdot}{\Omega}$ and $\Ltwonorm{\cdot}{\Omega}$, respectively. Using algebraic identities, the fact that $\director$ is of unit length, and integrating the above density function, the total free elastic energy for a domain $\Omega$ is
\begin{align*}
\int_{\Omega} w_F \diff{V} =& \frac{1}{2}(K_1-K_2-K_4) \Ltwonorm{\diverg \director}{\Omega}^2 +\frac{1}{2}K_3\Ltwoinnerndim{\vec{Z} \curl \director}{\curl \director}{\Omega}{3} \nonumber \\
& +\frac{1}{2}(K_2+K_4) \big(\Ltwoinnerndim{\nabla n_1}{\frac{\partial \director}{\partial x}}{\Omega}{3} +\Ltwoinnerndim{\nabla n_2}{\frac{\partial \director}{\partial y}}{\Omega}{3}+\Ltwoinnerndim{\nabla n_3}{\frac{\partial \director}{\partial z}}{\Omega}{3} \big).
\end{align*} 

For the special case of full Dirichlet boundary conditions, we consider a fixed director $\director$ at each point on the boundary of $\Omega$. Considering the integration carried out on the terms in \eqref{FrankOseenFree},
\begin{align} \label{stronganchoringdivthm}
&\frac{1}{2}(K_2+K_4) \int_{\Omega} \diverg[(\director \cdot \vec{\nabla}) \director - (\diverg \director) \director] \diff{V} \nonumber \\
&\qquad \qquad \qquad = \frac{1}{2}(K_2+K_4)\int_{\partial \Omega} [(\director \cdot \vec{\nabla}) \director - (\diverg \director) \director] \cdot \mathbf{\nu} \diff{S},
\end{align}
by the divergence theorem. Further, since $\director$ is fixed along $\partial \Omega$, the energy contributed by $\director$ on the boundary is constant regardless of the configuration of $\director$ on the interior of $\Omega$. Thus, in the minimization to follow, the energy contribution from this term is ignored. For this reason, \eqref{stronganchoringdivthm} is often referred to as a null Lagrangian \cite{Virga1}. Note that the above identity is also applicable to a rectangular domain with mixed Dirichlet and periodic boundary conditions. Such a domain will be considered in numerical experiments below.

A number of methods involving computation of liquid crystal equilibria or dynamics utilize the so called one-constant approximation that $K_1=K_2=K_3$ and $K_4 = 0$ \cite{Ramage1, Liu1, Stewart1, Cohen1}, in order to significantly simplify the free elastic energy density to
\begin{equation*}
\hat{w}_F = \frac{1}{2}K_1 \vert \nabla \director \vert^2, \text{ where } \vert \nabla \director \vert^2 = \sum_{i,j =1}^3 \left ( \pd{n_i}{x_j} \right)^2.
\end{equation*}
This expression for the free elastic energy density is more amenable to theoretical development but ignores significant physical characteristics \cite{Lee1, Atherton2}. The following method is derived without such an assumption.

This paper extends the approach of \cite{Emerson1} to consider electric fields. In the presence of an electric field, the free energy in a liquid crystal sample is directly affected. This interaction is strongly coupled as nematic polarization and electric displacement, in turn, affect the original electric field. The coupling is captured by an auxiliary term added to the Frank-Oseen equations such that the total system free energy has the form
\begin{equation}
\int_{\Omega} \big ( w_F - \frac{1}{2} \vec{D} \cdot \vec{E} \big ) \diff{V} \label{generalelectricfunctional},
\end{equation}
where $\vec{D}$ is the electric displacement vector induced by polarization and $\vec{E}$ is the local electric field \cite{deGennes1}. The electric displacement vector is written $\vec{D} = \epsilon_0 \epsilon_{\perp} \vec{E} + \epsilon_0 \epsilon_a(\director \cdot \vec{E})\director$. 
Here, $\epsilon_0>0$ is the permittivity of free space constant. The dielectric anisotropy constant is $\epsilon_a = \epsilon_{\parallel} - \epsilon_{\perp}$, where the constant variables $\epsilon_{\parallel} > 0$ and $\epsilon_{\perp} > 0$ represent the parallel and perpendicular dielectric permittivity, respectively, specific to the liquid crystal. If $\epsilon_a > 0$, the director is attracted to parallel alignment with the electric field, and if $\epsilon_a<0$, the director tends to align perpendicular to $\vec{E}$. Thus, 
\begin{equation*}
\vec{D} \cdot \vec{E}= \epsilon_0\epsilon_{\perp} \vec{E} \cdot \vec{E} + \epsilon_0\epsilon_a (\director \cdot \vec{E})^2.
\end{equation*}
Therefore, Equation \eqref{generalelectricfunctional} is expanded as
\begin{align}
\int_{\Omega} \big ( w_F - \frac{1}{2} \vec{D} \cdot \vec{E} \big ) \diff{V} =& \int_{\Omega} w_F \diff{V} - \frac{1}{2} \epsilon_0\epsilon_{\perp}\Ltwoinnerndim{\vec{E}}{\vec{E}}{\Omega}{3}- \frac{1}{2} \epsilon_0 \epsilon_a \Ltwoinner{\director \cdot \vec{E}}{\director \cdot \vec{E}}{\Omega} \label{ElectricSystemEnergy}.
\end{align}
The addition of the electric field not only increases the complexity of the functional, it introduces an inherent saddle-point structure into the equilibria for the liquid crystal samples. Energy minima are those that minimize the contribution of the free elastic energy, while maximizing the negative contribution of the electric field terms. Moreover, the relevant Maxwell's equations for a static electric field, $\diverg \vec{D} = 0$ and $\curl \vec{E} = \vec{0}$,
known as Gauss' and Faraday's laws, respectively, must be satisfied.

\section{Free Energy Minimization} \label{freeenergymin}

In \cite{Emerson1}, a general approach for computing the equilibrium state for $\director$ is derived. We apply this methodology to the augmented elastic-electric free energy. The equilibrium state corresponds to the configuration which minimizes the system free energy subject to the local constraint that $\director$ is of unit length throughout the sample volume, $\Omega$. That is, the minimizer must satisfy $\director \cdot \director = 1$ pointwise throughout the volume. In light of the necessary Maxwell equations and the fact that we are considering static fields, we reformulate the system energy in \eqref{ElectricSystemEnergy} using an electric potential function, $\phi$, such that $\vec{E} = - \nabla \phi$, and define the functional to be minimized as
\begin{align} \label{functional2}
\mathcal{F}_1(\director, \phi) &= (K_1-K_2-K_4) \Ltwonorm{\diverg \director}{\Omega}^2 + K_3\Ltwoinnerndim{\vec{Z} \curl \director}{\curl \director}{\Omega}{3} \nonumber \\
& \qquad + (K_2+K_4) \big(\Ltwoinnerndim{\nabla n_1}{\frac{\partial \director}{\partial x}}{\Omega}{3} + \Ltwoinnerndim{\nabla n_2}{\frac{\partial \director}{\partial y}}{\Omega}{3}+ \Ltwoinnerndim{\nabla n_3}{\frac{\partial \director}{\partial z}}{\Omega}{3} \big) \nonumber \\
& \qquad - \epsilon_0\epsilon_{\perp}\Ltwoinnerndim{\nabla \phi}{\nabla \phi}{\Omega}{3} - \epsilon_0 \epsilon_a \Ltwoinner{\director \cdot \nabla \phi}{\director \cdot \nabla \phi}{\Omega}.
\end{align}
Using a potential function guarantees that Faraday's law is trivially satisfied. Furthermore, it is not difficult to show that Gauss' law is satisfied at the minimum of the above functional.

In the presence of full Dirichlet boundary conditions or a rectangular domain with mixed Dirichlet and periodic boundary conditions, the functional to be minimized is significantly simplified to
\begin{align}
\mathcal{F}_2(\director, \phi) &= K_1 \Ltwonorm{\diverg \director}{\Omega}^2 + K_3\Ltwoinnerndim{\vec{Z} \curl \director}{\curl \director}{\Omega}{3} \nonumber \\
&\qquad - \epsilon_0\epsilon_{\perp}\Ltwoinnerndim{\nabla \phi}{\nabla \phi}{\Omega}{3}- \epsilon_0 \epsilon_a \Ltwoinner{\director \cdot \nabla \phi}{\director \cdot \nabla \phi}{\Omega}, \label{functional3}
\end{align} 
by the application of \eqref{stronganchoringdivthm}. However, the functional still contains nonlinear terms introduced by, for instance, the presence of $\vec{Z} = \vec{Z}(\director)$. 

We proceed with the functional in \eqref{functional2} in building a framework for minimization under general boundary conditions. However, in the treatment of existence and uniqueness theory, we assume the application of full Dirichlet or mixed Dirichlet and periodic boundary conditions and, therefore, utilize the simplified form in \eqref{functional3}.

As done in \cite{Emerson1}, we consider the spaces
\begin{align*}
\Hdiv{\Omega} &= \{\vec{v} \in L^2(\Omega)^3 : \diverg \vec{v} \in L^2(\Omega) \}, \\
\Hcurl{\Omega} &= \{ \vec{v} \in L^2(\Omega)^3 : \curl \vec{v} \in L^2(\Omega)^3 \}.
\end{align*}
Define 
\begin{equation*}
\Hdc= \{ \vec{v} \in \Hdiv{\Omega} \cap \Hcurl{\Omega} : B(\vec{v}) = \bar{\vec{g}} \},
\end{equation*}
 with norm $\Hdcnorm{\vec{v}}{\Omega}^2 = \Ltwonormndim{\vec{v}}{\Omega}{3}^2 + \Ltwonorm{\diverg \vec{v}}{\Omega}^2 + \Ltwonormndim{\curl \vec{v}}{\Omega}{3}^2$ and appropriate boundary conditions $B(\vec{v})=\bar{\vec{g}}$. Further, let $\Hdcnot = \{ \vec{v} \in \Hdiv{\Omega} \cap \Hcurl{\Omega} : B(\vec{v}) = \vec{0} \}$. Let
\begin{align*}
\Honeb{\Omega} = \{f \in \Hone{\Omega} : B_1(f) = g \},
\end{align*}
where $\Hone{\Omega}$ represents the classical Sobolev space and $B_1(f) = g$ is an appropriate boundary condition expression for $\phi$. Finally,  denote the unit sphere as $\mathcal{S}^2$. Using Functional \eqref{functional2}, the desired minimization becomes
\begin{equation*} \label{minioversphere}
\director_0, \phi_0 = \argmin_{\director, \phi \in \left (\mathcal{S}^2 \cap \Hdc \right) \times \Honeb{\Omega}}  \mathcal{F}_1(\director, \phi).
\end{equation*}

\subsection{First-Order Continuum Optimality Conditions} \label{electricfirstorderconditions}

Since $\director$ must be of unit length, it is natural to employ a Lagrange multiplier approach. This length requirement represents a pointwise equality constraint such that $\ltwoinner{\director}{\director} - 1 = 0$.
Thus, following general constrained optimization theory \cite{Luenberger1}, define the Lagrangian
\begin{align*}
\mathcal{L}(\director, \phi, \lambda) &= \mathcal{F}_1(\director, \phi) + \int_{\Omega} \lambda(\vec{x})(\ltwoinner{\director}{\director}-1) \diff{V},
\end{align*} 
where $\lambda \in \Ltwo{\Omega}$. In order to minimize \eqref{functional2}, we compute the G\^{a}teaux derivatives of $\mathcal{L}$ with respect to $\director$, $\phi$, and $\lambda$ in the directions $\vec{v} \in \Hdcnot$, $\psi \in \Honebnot{\Omega}$, and $\gamma \in L^2(\Omega)$, respectively. Hence, necessary continuum first-order optimality conditions are derived as
\begin{align*}
\lagdivn &= \frac{\partial}{\partial \director} \mathcal{L}(\director, \phi, \lambda) [\vec{v}] =0, & & \forall \vec{v} \in \Hdcnot, \\
\lagdivphi &= \frac{\partial}{\partial \phi} \mathcal{L}(\director, \phi, \lambda) [\psi] =0, & & \forall \psi \in \Honebnot{\Omega}, \\
\lagdivlam &= \frac{\partial}{\partial \lambda} \mathcal{L}(\director, \phi, \lambda) [\gamma] =0,& & \forall \gamma \in L^2(\Omega).
\end{align*}
Computing these derivatives yields the variational system
\begin{align*}
\lagdivn &= 2(K_1-K_2-K_4)\Ltwoinner{\diverg \director}{\diverg \vec{v}}{\Omega} + 2K_3\Ltwoinnerndim{\vec{Z} \curl \director}{\curl \vec{v}}{\Omega}{3} \nonumber  \\
& \qquad+ 2(K_2-K_3)\Ltwoinner{\director \cdot \curl \director}{\vec{v} \cdot \curl \director}{\Omega} + 2(K_2+K_4)\big(\Ltwoinnerndim{\nabla n_1}{\pd{\vec{v}}{x}}{\Omega}{3} \nonumber \\
& \qquad + \Ltwoinnerndim{\nabla n_2}{\pd{\vec{v}}{y}}{\Omega}{3} +\Ltwoinnerndim{\nabla n_3}{\pd{\vec{v}}{z}}{\Omega}{3} \big)-2 \epsilon_0 \epsilon_a \Ltwoinner{\director \cdot \nabla \phi}{\vec{v} \cdot \nabla \phi}{\Omega} \nonumber \\
& \qquad + 2 \int_{\Omega} \lambda \ltwoinner{\director}{\vec{v}} \diff{V} = 0, \qquad \hspace{4.45cm} \forall \vec{v} \in \Hdcnot, \\
\lagdivphi &= -2 \epsilon_0 \epsilon_{\perp}\Ltwoinnerndim{\nabla \phi}{\nabla \psi}{\Omega}{3} - 2 \epsilon_0 \epsilon_a \Ltwoinner{\director \cdot \nabla \phi}{\director \cdot \nabla \psi}{\Omega} = 0, \hspace{1.6cm} \forall \psi \in \Honebnot{\Omega}, \\
\lagdivlam &= \int_{\Omega} \gamma(\ltwoinner{\director}{\director} -1) \diff{V}=0, \hspace{5.2cm} \forall \gamma \in \Ltwo{\Omega}. 
\end{align*}

Note that $\lagdivphi = 0$, in the system above, is, in fact, the weak form of Gauss' law. Therefore, at the functional minimum both Gauss' and Faraday's laws are satisfied.

\subsection{Nonlinearities and Newton Linearization} \label{newtonstepssection}

The system above is nonlinear; therefore, Newton iterations are employed by computing a generalized first-order Taylor series expansion, requiring computation of the Hessian \cite{Benzi1, Nocedal1}. Let $\kdirector$, $\kphi$, and $\lambda_k$ be the current approximations for $\director$, $\phi$, and $\lambda$, respectively. Additionally, let $\ddirector= \director_{k+1} - \kdirector$, $\dphi = \phi_{k+1} - \kphi$, and $\dlambda = \lambda_{k+1}-\klambda$ be updates to the current approximations that we seek to compute. Then, the Newton iterations are denoted
\begin{equation} \label{ElectricPotentialNewtonSystem}
\left [ \begin{array}{c c c}
\mathcal{L}_{\director \director} & \mathcal{L}_{\director \phi} & \mathcal{L}_{\director \lambda} \\
\mathcal{L}_{\phi \director} & \mathcal{L}_{\phi \phi} & \mathcal{L}_{\phi \lambda} \\
\mathcal{L}_{\lambda \director} & \mathcal{L}_{\lambda \phi}& \mathcal{L}_{\lambda \lambda}
\end{array} \right] 
\left [ \begin{array}{c}
\ddirector \\
\dphi \\
\dlambda
\end{array} \right] = -
\left[ \begin{array}{c}
\mathcal{L}_{\director} \\
\mathcal{L}_{\phi} \\
\mathcal{L}_{\lambda} 
\end{array} \right],
\end{equation}
where each of the system components are evaluated at $\kdirector$, $\kphi$, and $\klambda$. The matrix-vector multiplication indicates the direction that the derivatives in the Hessian are taken. For instance, $\mathcal{L}_{\lambda \director}[\gamma] \cdot \ddirector = \pd{ }{\director} \left( \mathcal{L}_{\lambda} (\kdirector, \klambda)[\gamma] \right)[\ddirector]$, where the partials indicate G\^{a}teaux derivatives in the respective variables. Note that $\mathcal{L}_{\lambda \lambda} = \mathcal{L}_{\lambda \phi} =  \mathcal{L}_{\phi \lambda}=0$. Hence, the Hessian in \eqref{ElectricPotentialNewtonSystem} simplifies to a saddle-point matrix, which poses unique difficulties for the efficient computation of the solution to the resulting linear system. Such structures commonly appear in constrained optimization and other settings; for a comprehensive overview of discrete saddle-point problems see \cite{Benzi2}. Here, we focus only on the linearization step rather than the underlying linear solvers. An efficient iterative solver is discussed below. Considering the other six components of the Hessian, the derivatives involving $\lambda$ are
\begin{align*}
\mathcal{L}_{\lambda \director}[\gamma] \cdot \ddirector = 2 \int_{\Omega} \gamma \ltwoinner{\kdirector}{\ddirector} \diff{V}, & & \mathcal{L}_{\director \lambda}[\vec{v}] \cdot \dlambda= 2 \int_{\Omega} \dlambda \ltwoinner{\kdirector}{\vec{v}} \diff{V}.
\end{align*}
The second order terms involving $\phi$ are
\begin{align*}
\mathcal{L}_{\phi \phi}[\psi] \cdot \dphi &= -2 \epsilon_0 \epsilon_{\perp} \Ltwoinnerndim{\nabla \dphi}{\nabla \psi}{\Omega}{3} - 2 \epsilon_0 \epsilon_a \Ltwoinner{\kdirector \cdot \nabla \dphi}{\kdirector \cdot \nabla \psi}{\Omega}, \\
\mathcal{L}_{\phi \director}[\psi] \cdot \ddirector &=-2\epsilon_0 \epsilon_a \Ltwoinner{\kdirector \cdot \nabla \kphi}{\ddirector \cdot \nabla \psi}{\Omega} -2\epsilon_0 \epsilon_a \Ltwoinner{\ddirector \cdot \nabla \phi_k}{\kdirector \cdot \nabla \psi}{\Omega}, \\
\mathcal{L}_{\director \phi}[\vec{v}] \cdot \dphi &= - 2\epsilon_0 \epsilon_a \Ltwoinner{\kdirector \cdot \nabla \kphi}{\vec{v} \cdot \nabla \dphi}{\Omega} - 2\epsilon_0 \epsilon_a\Ltwoinner{\kdirector \cdot \nabla \dphi}{ \vec{v} \cdot \nabla \phi_k}{\Omega}.
\end{align*}
Finally, the second order derivative with respect to $\director$ is 
\begin{align*}
\mathcal{L}_{\director \director}[\vec{v}] \cdot \ddirector &= 2(K_1 - K_2 - K_4)\Ltwoinner{\diverg \ddirector}{\diverg \vec{v}}{\Omega} + 2K_3 \Ltwoinnerndim{\vec{Z}(\kdirector) \curl \ddirector}{\curl \vec{v}}{\Omega}{3} \nonumber \\
& \qquad + 2(K_2-K_3) \Big(\Ltwoinner{\ddirector \cdot \curl \vec{v}}{\kdirector \cdot \curl \kdirector}{\Omega} \nonumber \\
& \qquad +\Ltwoinner{\kdirector \cdot \curl \vec{v}}{\ddirector \cdot \curl \kdirector}{\Omega} + \Ltwoinner{\kdirector \cdot \curl \kdirector}{\vec{v} \cdot \curl \ddirector}{\Omega} \nonumber \\
& \qquad + \Ltwoinner{\kdirector \cdot \curl \ddirector}{\vec{v} \cdot \curl \kdirector}{\Omega} + \Ltwoinner{\ddirector \cdot \curl \kdirector}{\vec{v} \cdot \curl \kdirector}{\Omega}\Big) \nonumber \\
& \qquad + 2(K_2+K_4) \big( \Ltwoinnerndim{\nabla \delta n_1}{\pd{\vec{v}}{x}}{\Omega}{3} +\Ltwoinnerndim{\nabla \delta n_2}{\pd{\vec{v}}{y}}{\Omega}{3} + \Ltwoinnerndim{\nabla \delta n_3}{\pd{\vec{v}}{z}}{\Omega}{3} \big) \nonumber \\
& \qquad - 2 \epsilon_0 \epsilon_a \Ltwoinner{\ddirector \cdot \nabla \kphi}{\vec{v} \cdot \nabla \kphi}{\Omega}+ 2\int_{\Omega} \klambda \ltwoinner{\ddirector}{\vec{v}} \diff{V}.
\end{align*}

Completing \eqref{ElectricPotentialNewtonSystem} with the above Hessian computations yields a linearized variational system. For these iterations, we compute $\ddirector$, $\dphi$, and $\dlambda$ satisfying \eqref{ElectricPotentialNewtonSystem} for all $\vec{v} \in \Hdcnot$, $\psi \in \Honebnot{\Omega}$, and $\gamma \in L^2(\Omega)$ with the current approximations $\kdirector$, $\kphi$, and $\klambda$. While they typically improve robustness and efficiency, we do not consider the use of line searches or trust regions in the work presented here, leaving this for future work. If we are considering a system with Dirichlet boundary conditions, as described above, we eliminate the $(K_2 + K_4)$ terms from \eqref{ElectricPotentialNewtonSystem}. This produces a simplified, but non-trivial, linearization.

\section{Well-Posedness of the Discrete Systems} \label{wellposedhessian}

Performing the outlined Newton iterations necessitates solving the above linearized systems for the update functions $\ddirector$, $\dphi$, and $\dlambda$. Finite elements are used to numerically approximate these updates as $\ddirector_h$, $\dphi_h$, and $\dlambda_h$. Throughout this section, we assume that full Dirichlet boundary conditions are enforced for $\director$ and $\phi$. However, the following theory is also applicable for a rectangular domain with mixed Dirichlet and periodic boundary conditions. Such a domain is considered for the numerical experiments presented herein. 

We write the bilinear form defined by $-\mathcal{L}_{\phi \phi}[\psi] \cdot \dphi$ as $c(\dphi, \psi)=\epsilon_0 \epsilon_{\perp} \Ltwoinnerndim{\nabla \dphi}{\nabla \psi}{\Omega}{3} + \epsilon_0 \epsilon_a \Ltwoinner{\kdirector \cdot \nabla \dphi}{\kdirector \cdot \nabla \psi}{\Omega}$ and the form associated with $\mathcal{L}_{\lambda \director}[\gamma] \cdot \ddirector$ as $b(\ddirector, \gamma)$. Further, we decompose the bilinear form defined by $\mathcal{L}_{\director \director}[\vec{v}] \cdot \ddirector$ into a free elastic term, $\tilde{a}(\ddirector, \vec{v})$, and an electric component as
\begin{equation*}
a(\ddirector, \vec{v}) = \tilde{a}(\ddirector, \vec{v}) - \epsilon_0 \epsilon_a \Ltwoinner{\ddirector \cdot \nabla \kphi}{\vec{v} \cdot \nabla \kphi}{\Omega}.
\end{equation*}
\begin{lemma} \label{cpositivedefinite}
Let $\Omega$ be a connected, open, bounded domain. If $\epsilon_a \geq 0$, then $c(\dphi, \psi)$ is a coercive bilinear form. For $\epsilon_a < 0$, if $\ltwonorm{\kdirector}^2 \leq \beta < \epsilon_{\perp}/\vert \epsilon_a \vert$, then $c(\dphi, \psi)$ is a coercive bilinear form.
\end{lemma}
\begin{proof}
The proof is split into two cases.
\begin{caseof}
\case{$\epsilon_a \geq 0$.}
Note that $\dphi, \psi \in \Honebnot{\Omega}$, with homogeneous Dirichlet boundary conditions. By the classical Poincar\'{e}-Friedrichs' inequality, there exists a $C_1 > 0$ such that for all $\xi \in \Honenot{\Omega}$, $\Ltwonormndim{\xi}{\Omega}{3}^2 \leq C_1 \Ltwonormndim{\nabla \xi}{\Omega}{3}^2$. Therefore,
\begin{equation*}
\Honenorm{\xi}{\Omega}^2 \leq (C_1+1) \Ltwonormndim{\nabla \xi}{\Omega}{3}^2.
\end{equation*}
This implies that, for $\xi \neq 0$,
\begin{align*}
c(\xi, \xi) &= \epsilon_0 \epsilon_{\perp} \Ltwoinnerndim{\nabla \xi}{\nabla \xi}{\Omega}{3} + \epsilon_0 \epsilon_a \Ltwoinner{\kdirector \cdot \nabla \xi}{\kdirector \cdot \nabla \xi}{\Omega} \nonumber \\
& \geq \frac{\epsilon_0 \epsilon_{\perp}}{C_1+1} \Honenorm{\xi}{\Omega}^2 > 0.
\end{align*}
\case{$\epsilon_a < 0$.}
Observe that pointwise,
\begin{equation*}
(\kdirector \cdot \nabla \xi)^2 \leq \ltwonorm{\kdirector}^2 \ltwonorm{\nabla \xi}^2 \leq \beta \ltwonorm{\nabla \xi}^2.
\end{equation*}
This implies that $\Ltwoinner{\kdirector \cdot \nabla \xi}{\kdirector \cdot \nabla \xi}{\Omega} \leq \beta \Ltwoinnerndim{\nabla \xi}{\nabla \xi}{\Omega}{3}$. Therefore,
\begin{equation*}
c(\xi, \xi) \geq \epsilon_0 (\epsilon_{\perp} - \beta \vert \epsilon_a \vert) \Ltwoinnerndim{\nabla \xi}{\nabla \xi}{\Omega}{3}.
\end{equation*}
Recall that $\epsilon_{\perp} > 0$. Therefore, $\beta < \epsilon_{\perp}/\vert \epsilon_a \vert$ implies that $\epsilon_{\perp} - \beta \vert \epsilon_a \vert > 0$. Thus, again applying the Poincar\'{e}-Friedrichs' inequality above for $\xi \neq 0$,
\begin{equation*}
c(\xi, \xi) \geq \frac{\epsilon_0(\epsilon_{\perp} - \beta \vert \epsilon_a \vert)}{C_1+1} \Honenorm{\xi}{\Omega}^2 > 0.
\end{equation*}
\end{caseof}
In either case, $c(\cdot, \cdot)$ is a coercive bilinear form.
\end{proof}

There are a number of discretization space triples commonly used to discretize systems such as the one defined in \eqref{ElectricPotentialNewtonSystem}, including equal order or mixed finite elements. Discretizing the Hessian in \eqref{ElectricPotentialNewtonSystem} with finite elements leads to the $3 \times 3$ block matrix
\begin{equation} \label{decomphessian}
M = \left [ \begin{array}{c c c}
A & B_1 & B_2 \\
B_1^T & -\tilde{C} & \vec{0} \\
B_2^T & \vec{0} & \vec{0}
\end{array} \right ].
\end{equation}
\begin{lemma} \label{invertibilityM}
Under the assumptions in Lemma \ref{cpositivedefinite}, if the bilinear forms $a(\cdot, \cdot)$ and $b(\cdot, \cdot)$, defined above, are coercive and weakly coercive, respectively, on the relevant discrete spaces,  the matrix in \eqref{decomphessian} is invertible.
\end{lemma}
\begin{proof}
Denoting $B = \left[ \begin{array}{c c} B_1 & B_2 \end{array} \right]$ (where $B_2$ is associated with $b(\cdot, \cdot)$), and $C = \left[ \begin{array}{c c} \tilde{C} & \vec{0} \\ \vec{0} & \vec{0} \end{array} \right]$, the matrix in \eqref{decomphessian} is written as
\begin{equation*}
\left [ \begin{array}{c c}
A & B \\
B^T & -C
\end{array}
\right].
\end{equation*}
By assumption, $a(\cdot, \cdot)$ is coercive, and it is clearly symmetric \cite{Emerson1}. Therefore, the associated discretization block, $A$, is symmetric positive definite. By Lemma \ref{cpositivedefinite}, $\tilde{C}$ is symmetric positive definite, and, therefore, $-C$ is symmetric negative semi-definite. Therefore, by \cite[Theorem 3.1]{Benzi2}, if $\ker{C} \cap \ker{B} = \{\vec{0} \}$, then the matrix in \eqref{decomphessian} is invertible. Observe that
\begin{equation*}
\left[ \begin{array}{c c}
\tilde{C} & \vec{0} \\
\vec{0} & \vec{0}
\end{array} \right]
\left[ \begin{array}{c}
\vec{y} \\
\vec{z}
\end{array} \right] = 
\left [ \begin{array}{c}
\tilde{C} \vec{y} \\
\vec{0}
\end{array} \right] = \vec{0}
\end{equation*}
if and only if $\vec{y} = \vec{0}$. Then, if $\left[ \begin{array}{c c} \vec{y} & \vec{z} \end{array} \right]^T \in \ker{C} \cap \ker{B}$, $\vec{y} = \vec{0}$. However, note that
\begin{equation*}
\left [ \begin{array}{c c}
B_1 & B_2
\end{array} \right]
\left [ \begin{array}{c}
\vec{0} \\
\vec{z}
\end{array} \right] = B_2 \vec{z}.
\end{equation*}
Since $b(\cdot, \cdot)$ is weakly coercive, $B_2 \vec{z} = \vec{0}$ if and only if $\vec{z}=\vec{0}$. So $\ker{C} \cap \ker{B}=\{\vec{0}\}$.
\end{proof}
 
For the remainder of the paper, let $C_{\phi} = \displaystyle{\sup_{\vec{x} \in \Omega} \vert \nabla \kphi \vert}$. Furthermore, for $\triangulation$, a quadrilateral subdivision of $\Omega$, let $Q_p$ denote piecewise $C^0$ polynomials of degree $p \geq 1$ and $P_0$ denote the space of piecewise constants. Define a bubble space 
\begin{equation*}
V_h^b = \{\vec{v} \in C_c(\Omega)^3 : \vec{v}|_{T} =a_T b_T \kdirector |_T, \forall T \in \triangulation \},
\end{equation*}
where $C_c(\Omega)$ denotes the space of compactly supported continuous functions on $\Omega$, $b_T$ is the biquadratic bubble function \cite{Mourad1} that vanishes on $\partial T \in \triangulation$, and $a_T$ is a constant coefficient associated with $b_T$. Then the discretization spaces considered for $\dlambda$ and $\ddirector$, respectively, are
\begin{align}
\Pi_h &= P_0, \label{pispace} \\
V_h &= \{ \vec{v} \in Q_m \times Q_m \times Q_m \oplus V_h^b : \vec{v} = \vec{0} \text{ on } \partial \Omega \}. \label{vspace}
\end{align}
Note that Lemma 3.12 in \cite{Emerson1} uses these spaces to show that $b(\cdot, \cdot)$ is weakly coercive. The above lemma now allows for the formulation of the following theorem using the discrete spaces above.


\begin{theorem} \label{elecSysInvThm}
Under the assumptions of Lemmas $3.7$ or $3.8$ in \cite{Emerson1}, for $\kappa =1$ or $\kappa$ satisfying the small data assumptions in \cite[Lemma 3.8]{Emerson1}, respectively, let $\alpha_0 > 0$ be such that $\tilde{a}(\vec{v}, \vec{v}) \geq \alpha_0 \Hdcnorm{\vec{v}}{\Omega}^2$. With the assumptions of Lemma 3.12 in \cite{Emerson1} and those of Lemma \ref{cpositivedefinite}, if $\epsilon_a \leq 0$ or $(\alpha_0 - \epsilon_0 \epsilon_a C^2_{\phi}) > 0$, then the matrix defined by \eqref{decomphessian} is invertible. 
\end{theorem}
\begin{proof}
If $\kappa =1$, Lemma 3.7 in \cite{Emerson1} implies that such an $\alpha_0 > 0$ exists. Similarly, if $\kappa$ satisfies the small data assumptions in \cite[Lemma 3.8]{Emerson1}, then such an $\alpha_0 > 0$ also exists. If $\epsilon_a \leq 0$, clearly this implies that $a(\cdot, \cdot)$ is coercive. For $\epsilon_a > 0$, note that
\begin{align}
\Ltwoinner{\vec{v} \cdot \nabla \kphi}{\vec{v} \cdot \nabla \kphi}{\Omega} = \int_{\Omega} (\vec{v} \cdot \nabla \kphi)^2 \diff{V} &\leq \int_{\Omega} \vert \vec{v} \vert^2 \vert \nabla \kphi \vert^2 \diff{V} \nonumber \\
&\leq C_{\phi}^2 \int_{\Omega} \vert \vec{v} \vert^2 \diff{V} \nonumber \\
&\leq C_{\phi}^2 \Hdcnorm{\vec{v}}{\Omega}^2 \label{boundvgradphi}.
\end{align}
Hence,
\begin{equation} \label{bounde0posea}
\vert \epsilon_0 \epsilon_a \Ltwoinner{\vec{v} \cdot \nabla \kphi}{\vec{v} \cdot \nabla \kphi}{\Omega} \vert \leq \epsilon_0 \epsilon_a C_{\phi}^2 \Hdcnorm{\vec{v}}{\Omega}^2.
\end{equation}
Therefore,
\begin{align*}
a(\vec{v}, \vec{v}) &\geq \alpha_0 \Hdcnorm{\vec{v}}{\Omega}^2 - \epsilon_0 \epsilon_a C_{\phi}^2 \Hdcnorm{\vec{v}}{\Omega}^2 \\
& = (\alpha_0 - \epsilon_0 \epsilon_a C_{\phi}^2) \Hdcnorm{\vec{v}}{\Omega}^2. 
\end{align*}
Thus, if $(\alpha_0 - \epsilon_0 \epsilon_a C^2_{\phi}) > 0$, $a(\cdot, \cdot)$ is coercive.

Finally, Lemma 3.12 in \cite{Emerson1} asserts that $b(\cdot, \cdot)$ is weakly coercive. Hence, Lemma \ref{invertibilityM} implies that $M$, as defined in \eqref{decomphessian}, is invertible.
\end{proof}

Theorem \ref{elecSysInvThm} implies that no additional inf-sup condition for $\phi$ is necessary to guarantee uniqueness of the solution to the system in \eqref{ElectricPotentialNewtonSystem}. Moreover, the discretization space for $\phi$ may be freely chosen without concern for stability.

\section{Flexoelectric Augmentation} \label{flexoaugmentation}

Flexoelectricity is a property demonstrated by certain dielectric materials, including liquid crystals. It is a spontaneous polarization of the liquid crystal induced by present curvature; it is caused by shape asymmetry of the constituent molecules of the liquid crystal material. The initial suggestion of this type of property in liquid crystals was introduced by Meyer \cite{Meyer1}. This phenomenon can, for instance, be useful in the conversion of mechanical energy to electrical energy via large deformations of the boundary containing a liquid crystal sample \cite{Jakli}. It can also play a significant role in determining the equilibrium states of liquid crystal samples with patterned surface boundaries. For example, it is an important effect in the bistable configuration of the Zenithal Bistable Device (ZBD) \cite{Davidson1}. 

The effect of flexoelectricity on the alignment of a liquid crystal bulk is modeled by an augmentation of the electric displacement vector $\vec{D}$, discussed above, and additional terms for the bulk free energy functional. The electric displacement vector is modified \cite{Elston1} such that
\begin{equation*} \label{flexodisplacement}
\vec{D} = \epsilon_0 \epsilon_{\perp} \vec{E} + \epsilon_0 \epsilon_a (\director \cdot \vec{E}) \director + \Pflex.
\end{equation*}
Following the notation and sign convention of Rudquist \cite{Rudquist1} we write
\begin{equation} \label{Pflexovector}
\Pflex = e_s \director (\diverg \director) + e_b (\director \times \curl \director),
\end{equation}
where $e_s$ and $e_b$ are material constants specific to the liquid crystal. It is also common in physics literature to denote these constants as $e_1$ and $e_3$ under a separate sign convention \cite{Elston1, deGennes1, Meyer1}.

As expressed in \cite{Elston1}, the free energy density due to the additional flexoelectric effects is
\begin{align} \label{flexodensity}
- \Pflex \cdot \vec{E}.
\end{align}
Therefore, using \eqref{Pflexovector} and \eqref{flexodensity}, the additional free energy contributed by flexoelectric polarization is given as
\begin{equation*}
- \int_{\Omega} e_s (\diverg \director)(\vec{E} \cdot \director) + e_b (\director \times \curl \director) \cdot \vec{E} \diff{V}.
\end{equation*}
Substituting an electric potential function, $\vec{E} = -\nabla \phi$, the flexoelectric free energy functional to be minimized is expressed,
\begin{align}
\mathcal{F}_3 (\director, \phi) &= \mathcal{F}_1(\director, \phi) + 2e_s \Ltwoinner{\diverg \director}{\director \cdot \nabla \phi}{\Omega} + 2e_b\Ltwoinnerndim{\director \times \curl \director}{\nabla \phi}{\Omega}{3} \label{flexoelectricfunctional}.
\end{align}
Note that the redefinition of $\vec{D}$ applies purely to the computation of Gauss' Law and does not change the electric energy in $\mathcal{F}_1(\director, \phi)$. As above, in the presence of full Dirichlet or mixed Dirichlet and periodic boundary conditions on a rectangular domain, the simplification in \eqref{stronganchoringdivthm} is applied to eliminate the $(K_2 + K_4)$ terms from \eqref{flexoelectricfunctional}. Additionally, note that the Maxwell's equations, $\diverg \vec{D} = 0$ and $\curl \vec{E} = \vec{0}$, must still be satisfied. As before, the use of the electric potential implies that Faraday's law is automatically satisfied, and it can be shown that a minimizing triple $(\director_{*}, \phi_{*}, \lambda_{*})$ for the extended functional, \eqref{flexoelectricfunctional}, satisfies Gauss' law in weak form.

\subsection{Flexoelectric System}

With the goal of minimizing $\mathcal{F}_3$ subject to the local unit length constraint, define the flexoelectric Lagrangian
\begin{equation}
\hat{\mathcal{L}}(\director, \phi, \lambda) = \mathcal{F}_3(\director, \phi) + \int_{\Omega} \lambda(\ltwoinner{\director}{\director}-1) \diff{V} \label{flexoLagrangian}.
\end{equation}
As in Section \ref{electricfirstorderconditions}, in order to minimize \eqref{flexoLagrangian}, G\^{a}teaux derivatives for $\hat{\mathcal{L}}(\director, \phi, \lambda)$ must be computed. Derivation of this variational system is identical to that of the simple electric conditions with the exception of the derivative calculations for the additional flexoelectric energy terms. Therefore, the complete flexoelectric variational system is
\begin{align*}
\hat{\mathcal{L}}_{\director}[\vec{v}] &= \lagdivn +2e_s\big( \Ltwoinner{\diverg \director}{\vec{v} \cdot \nabla \phi}{\Omega} + \Ltwoinner{\diverg \vec{v}}{\director \cdot \nabla \phi}{\Omega} \big) \nonumber \\
& \qquad + 2e_b\big( \Ltwoinnerndim{\director \times \curl \vec{v}}{\nabla \phi}{\Omega}{3} + \Ltwoinnerndim{\vec{v} \times \curl \director}{\nabla \phi}{\Omega}{3} \big) = 0, \hspace{1cm} \forall \vec{v} \in \Hdcnot, \\
\hat{\mathcal{L}}_{\phi}[\psi] &= \lagdivphi + 2e_s \Ltwoinner{\diverg \director}{\director \cdot \nabla \psi}{\Omega} +2e_b  \Ltwoinnerndim{\director \times \curl \director}{\nabla \psi}{\Omega}{3} = 0, \hspace{.5cm} \forall \psi \in \Honebnot{\Omega}, \\
\hat{\mathcal{L}}_{\lambda}[\gamma] &= \lagdivlam =0, \hspace{7.49cm} \forall \gamma \in \Ltwo{\Omega}. 
\end{align*}

Constructing the Newton iterations to address the nonlinearities, as above, yields a Newton linearization system with a saddle-point structure similar to that of the electric field case. Since the flexoelectric energy terms are first-order with respect to $\phi$ and do not depend of $\lambda$, many of the second order derivatives are the same as the simple electric case. On the other hand, the mixed partial derivatives involving $\phi$ contain additional terms,
\begin{align*}
\hat{\mathcal{L}}_{\phi \director}[\psi] \cdot \ddirector &= \mathcal{L}_{\phi \director}[\psi] \cdot \ddirector + 2e_s \big( \Ltwoinner{\diverg \ddirector}{\kdirector \cdot \nabla \psi}{\Omega} + \Ltwoinner{\diverg \kdirector}{\ddirector \cdot \nabla \psi}{\Omega} \big)\nonumber \\
& \qquad+ 2e_b \big( \Ltwoinnerndim{\kdirector \times \curl \ddirector}{\nabla \psi}{\Omega}{3}+ \Ltwoinnerndim{\ddirector \times \curl \kdirector}{\nabla \psi}{\Omega}{3} \big), \\
\hat{\mathcal{L}}_{\director \phi}[\vec{v}] \cdot \dphi &= \mathcal{L}_{\director \phi}[\vec{v}] \cdot \dphi + 2e_s \big( \Ltwoinner{\diverg \kdirector}{\vec{v} \cdot \nabla \dphi}{\Omega} + \Ltwoinner{\diverg \vec{v}}{\kdirector \cdot \nabla \dphi}{\Omega} \big) \nonumber \\
& \qquad+ 2e_b \big(\Ltwoinnerndim{\kdirector \times \curl \vec{v}}{\nabla \dphi}{\Omega}{3} + \Ltwoinnerndim{\vec{v} \times \curl \kdirector}{\nabla \dphi}{\Omega}{3} \big).
\end{align*}
Finally, the second order derivative with respect to $\director$ also contains additional terms,
\begin{align*}
\hat{\mathcal{L}}_{\director \director}[\vec{v}] \cdot \ddirector &= \mathcal{L}_{\director \director}[\vec{v}] \cdot \ddirector + 2e_s \big(\Ltwoinner{\diverg \ddirector}{\vec{v} \cdot \nabla \kphi}{\Omega} + \Ltwoinner{\diverg \vec{v}}{\ddirector \cdot \nabla \kphi}{\Omega} \big) \nonumber \\
& \qquad + 2e_b \big( \Ltwoinnerndim{\ddirector \times \curl \vec{v}}{\nabla \kphi}{\Omega}{3} + \Ltwoinnerndim{\vec{v} \times \curl \ddirector}{\nabla \kphi}{\Omega}{3} \big).
\end{align*}
Completing the system in \eqref{ElectricPotentialNewtonSystem} with the above Hessian and right hand side computations yields the flexoelectric linearized variational system.

\subsection{Well-Posedness of the Discrete Flexoelectric Systems} \label{flexoinvertibility}

As with the simple electric linearization, finite elements are used to numerically approximate the updates as $\ddirector_h$, $\dphi_h$, and $\dlambda_h$. For simplicity, throughout this section we assume that full Dirichlet boundary conditions are enforced for $\director$ and $\phi$. However, the theory is, as above, also applicable for a rectangular domain with mixed Dirichlet and periodic boundary conditions. As in the simple electric case, we define bilinear forms to represent relevant components of the computed Hessian. The bilinear forms associated with $-\hat{\mathcal{L}}_{\phi \phi}[\psi] \cdot \dphi$ and $\hat{\mathcal{L}}_{\lambda \director}[\gamma] \cdot \ddirector$ are denoted $c(\dphi, \psi)$ and $b(\ddirector, \gamma)$, respectively, and are identical to the corresponding components of the simple electric case above. We again decompose the bilinear form defined by $\hat{\mathcal{L}}_{\director \director}[\vec{v}] \cdot \ddirector$ into a free elastic term, $\tilde{a}(\ddirector, \vec{v})$, and a flexoelectric component as
\begin{align*}
a(\ddirector, \vec{v}) &= \tilde{a}(\ddirector, \vec{v}) - \epsilon_0 \epsilon_a \Ltwoinner{\ddirector \cdot \nabla \kphi}{\vec{v} \cdot \nabla \kphi}{\Omega} \\
&\qquad + e_s \big(\Ltwoinner{\diverg \ddirector}{\vec{v} \cdot \nabla \kphi}{\Omega} + \Ltwoinner{\diverg \vec{v}}{\ddirector \cdot \nabla \kphi}{\Omega} \big) \nonumber \\
&\qquad+ e_b \big( \Ltwoinnerndim{\ddirector \times \curl \vec{v}}{\nabla \kphi}{\Omega}{3}+ \Ltwoinnerndim{\vec{v} \times \curl \ddirector}{\nabla \kphi}{\Omega}{3} \big).
\end{align*}
Recalling that $C_{\phi} = \displaystyle{\sup_{\vec{x} \in \Omega} \vert \nabla \kphi \vert}$, we formulate the following lemma.
\begin{lemma} \label{flexoacoercivity}
Under the assumptions of Lemma 3.7 or 3.8 from \cite{Emerson1}, let $\alpha_0 > 0$ be such that $\tilde{a}(\vec{v}, \vec{v}) \geq \alpha_0 \Hdcnorm{\vec{v}}{\Omega}^2$. If $\epsilon_a \leq 0$ and $\alpha_0 > 2 C_{\phi} (\vert e_b \vert + \vert e_s \vert)$ or $\epsilon_a > 0$ and $\alpha_0 > \epsilon_0 \epsilon_a C_{\phi}^2 + 2 C_{\phi}(\vert e_b \vert  + \vert e_s \vert )$, then there exists an $\alpha_1 > 0$ such that $a(\vec{v}, \vec{v}) \geq \alpha_1 \Hdcnorm{\vec{v}}{\Omega}^2$.
\end{lemma}
\begin{proof}
The proof is split into two cases.
\begin{caseof}
\case{$\epsilon_a \leq 0$.}{Since $\epsilon_0>0$ and $\Ltwoinner{\vec{v} \cdot \nabla \kphi}{\vec{v} \cdot \nabla \kphi}{\Omega}$ is clearly positive definite,
\begin{equation}
\tilde{a} (\vec{v}, \vec{v}) - \epsilon_0 \epsilon_a \Ltwoinner{\vec{v} \cdot \nabla \kphi}{\vec{v} \cdot \nabla \kphi}{\Omega} \geq \alpha_0 \Hdcnorm{\vec{v}}{\Omega}^2.
\end{equation}
Note that
\begin{align*}
\vert 2 e_s \Ltwoinner{\diverg \vec{v}}{\vec{v} \cdot \nabla \kphi}{\Omega} \vert &\leq 2 \vert e_s \vert \Ltwonorm{\diverg \vec{v}}{\Omega} \Ltwonorm{\vec{v} \cdot \nabla \kphi}{\Omega} \\
& \leq 2 \vert e_s \vert \Hdcnorm{\vec{v}}{\Omega} \Ltwonorm{\vec{v} \cdot \nabla \kphi}{\Omega}.
\end{align*}
Furthermore, from \eqref{boundvgradphi},
\begin{align*}
 \Ltwonorm{\vec{v} \cdot \nabla \kphi}{\Omega}^2 \leq C^2_{\phi} \Hdcnorm{\vec{v}}{\Omega}^2.
\end{align*}
Hence,
\begin{align} \label{boundesnegea}
\vert 2 e_s \Ltwoinner{\diverg \vec{v}}{\vec{v} \cdot \nabla \kphi}{\Omega} \vert \leq 2 C_{\phi} \vert e_s \vert  \Hdcnorm{\vec{v}}{\Omega}^2.
\end{align}
Bounding the second relevant term,
\begin{align*}
\vert 2 e_b \Ltwoinner{\vec{v} \times \curl \vec{v}}{\nabla \kphi}{\Omega} \vert &\leq 2 \vert e_b \vert \vert \Ltwoinner{\vec{v}}{(\curl \vec{v}) \times \nabla \kphi}{\Omega} \vert \\
& \leq 2 \vert e_b \vert \Ltwonorm{\vec{v}}{\Omega} \Ltwonorm{(\curl \vec{v}) \times \nabla \kphi}{\Omega}.
\end{align*}
Pointwise,
\begin{equation*}
\vert (\curl \vec{v}) \times \nabla \kphi \vert^2 \leq \vert \curl \vec{v} \vert^2 \vert \nabla \kphi \vert^2.
\end{equation*}
Therefore,
\begin{align*}
\Ltwonorm{(\curl \vec{v}) \times \nabla \kphi}{\Omega}^2 = \int_{\Omega} \vert (\curl \vec{v}) \times \nabla \kphi \vert^2 \diff{V} &\leq \int_{\Omega} \vert \curl \vec{v} \vert^2 \vert \nabla \kphi \vert^2 \diff{V} \\
& \leq C_{\phi}^2 \int_{\Omega} \vert \curl \vec{v} \vert^2 \diff{V} \\
& \leq C_{\phi}^2 \Ltwonorm{\curl \vec{v}}{\Omega}^2 \leq C^2_{\phi} \Hdcnorm{\vec{v}}{\Omega}^2.
\end{align*}
Thus,
\begin{align}
\vert 2 e_b \Ltwoinner{\vec{v} \times \curl \vec{v}}{\nabla \kphi}{\Omega} \vert & \leq 2  C_{\phi} \vert e_b \vert \Ltwonorm{\vec{v}}{\Omega} \Ltwonorm{\curl \vec{v}}{\Omega} \nonumber \\
&\leq 2 C_{\phi} \vert e_b \vert \Hdcnorm{\vec{v}}{\Omega}^2. \label{boundebneqea}
\end{align}
Gathering the bounds in \eqref{boundesnegea}-\eqref{boundebneqea},
\begin{align*}
a(\vec{v}, \vec{v}) &\geq \alpha_0 \Hdcnorm{\vec{v}}{\Omega}^2 - 2 \vert e_b \vert C_{\phi} \Hdcnorm{\vec{v}}{\Omega}^2 - 2 \vert e_s \vert C_{\phi} \Hdcnorm{\vec{v}}{\Omega}^2 \\
&=(\alpha_0 - 2 C_{\phi} (\vert e_b \vert  + \vert e_s \vert)) \Hdcnorm{\vec{v}}{\Omega}^2.
\end{align*}
Then, set $\alpha_1 = \alpha_0 - 2C_{\phi} (\vert e_b \vert + \vert e_s \vert) > 0$.
}
\case{$\epsilon_a > 0$.}{In this case the additional term, $\Ltwoinner{\vec{v} \cdot \nabla \kphi}{\vec{v} \cdot \nabla \kphi}{\Omega}$, is important. Recall, from \eqref{bounde0posea}, that
\begin{equation} \label{bounde0posea2}
\vert \epsilon_0 \epsilon_a \Ltwoinner{\vec{v} \cdot \nabla \kphi}{\vec{v} \cdot \nabla \kphi}{\Omega} \vert \leq \epsilon_0 \epsilon_a C_{\phi}^2 \Hdcnorm{\vec{v}}{\Omega}^2.
\end{equation}
Employing the bounds in \eqref{boundesnegea}-\eqref{bounde0posea2},
\begin{align*}
a(\vec{v}, \vec{v}) &\geq \alpha_0 \Hdcnorm{\vec{v}}{\Omega}^2 - \epsilon_0 \epsilon_a C_{\phi}^2 \Hdcnorm{\vec{v}}{\Omega}^2 - 2 C_{\phi}(\vert e_b \vert  + \vert e_s \vert ) \Hdcnorm{\vec{v}}{\Omega}^2 \\
& = (\alpha_0 - (\epsilon_0 \epsilon_a C_{\phi}^2 + 2 C_{\phi}(\vert e_b \vert  + \vert e_s \vert ))) \Hdcnorm{\vec{v}}{\Omega}^2.
\end{align*}
Thus, let $\alpha_1 = \alpha_0 - (\epsilon_0 \epsilon_a C_{\phi}^2 + 2 C_{\phi}(\vert e_b \vert  + \vert e_s \vert ))>0$.
}
\end{caseof}
\end{proof}

When discretizing the flexoelectric linearization, the $3 \times 3$ saddle-point block structure,
\begin{equation} \label{flexodecomphessian}
\bar{M} = \left [ \begin{array}{c c c}
\bar{A} & \bar{B}_1 & B_2 \\
\bar{B}_1^T & -\tilde{C} & \vec{0} \\
B_2^T & \vec{0} & \vec{0}
\end{array} \right ],
\end{equation}
described in \eqref{decomphessian} resurfaces. Blocks $B_2$ and $\tilde{C}$ are identical to those in \eqref{decomphessian} as they are discretizations of the same bilinear forms in the simple electric case. Again, making use of the discretization spaces defined in \eqref{pispace} and \eqref{vspace} above, the following theorem holds.
\begin{theorem} \label{invertibilitybarM}
Under the assumptions of Lemma 3.12 in \cite{Emerson1} and Lemmas \ref{cpositivedefinite} and \ref{flexoacoercivity}, $\bar{M}$ is invertible.
\end{theorem}
\begin{proof}
Lemma 3.12 in \cite{Emerson1} implies that the bilinear form $b(\ddirector, \gamma)$, associated with $B_2$, is weakly coercive and Lemma \ref{flexoacoercivity} implies that $a(\ddirector, \vec{v})$ is coercive. Therefore, Lemma \ref{invertibilityM} implies that $\bar{M}$ is invertible.
\end{proof}

Therefore, as in the simple electric case above, Theorem \ref{invertibilitybarM} implies that no additional inf-sup condition for $\phi$ is necessary to guarantee uniqueness of the solution to the system in \eqref{ElectricPotentialNewtonSystem}, and the discretization space for $\phi$ may be freely chosen without concern for stability.

\section{Numerical Methodology} \label{nummethodology}

The algorithm to perform the minimizations discussed in previous sections has three stages and was developed in \cite{Emerson1} for the elastic case; see Algorithm \ref{algo}. The outermost phase is nested iteration (NI) \cite{McCormick1, Starke1}, which begins on a specified coarsest grid level. Newton iterations are performed on each grid, updating the current approximation after each step. The stopping criterion for the Newton iterations at each level is based on a specified tolerance for the current approximation's conformance to the first-order optimality conditions in the standard Euclidean $l_2$ norm. The resulting approximation is then interpolated to a finer grid. The current implementation performs uniform grid refinement after each set of Newton iterations.

The linear system for each Newton step has the anticipated saddle-point block structure, detailed in \eqref{decomphessian} and \eqref{flexodecomphessian}. For the numerical experiments, the matrices are inverted using a coupled multigrid approach with Vanka-type relaxation, discussed below, in order to approximately solve for the discrete updates $\ddirector_h$, $\dphi_h$, and $\dlambda_h$. Finally, an incomplete Newton correction is performed. That is, the new iterates are given by
\begin{equation} \label{corrections}
 \left [ \begin{array}{c}
 \director_{k+1} \\
 \phi_{k+1} \\
 \lambda_{k+1} 
 \end{array} \right ]
= \left [ \begin{array}{c} 
\kdirector \\
\kphi \\
\klambda \\
\end{array} \right ] + \omega
\left [ \begin{array}{c}
\ddirector_h \\
\dphi_h \\
\dlambda_h
\end{array} \right ],
 \end{equation}
where $\omega \leq 1$. This is to ensure relatively strict adherence to the constraint manifold, which is necessary for the invertibility discussed above. For this algorithm, $\omega$ is chosen to begin at $0.2$ on the coarsest grid and increases by $0.2$, to a maximum of $1$, after each grid refinement, so that as the approximation converges, larger Newton steps are taken. The grid management and discretizations are implemented using the deal.II finite-element library, which is an aggressively optimized and parallelized open-source library widely used in scientific computing \cite{BangerthHartmannKanschat2007, DealIIReference}. In the numerical tests to follow, $Q_2$--$Q_2$--$P_0$ discretizations are used to approximate $\ddirector_h$, $\dphi_h$, and $\dlambda_h$, respectively, on each grid. Note that these spaces differ slightly from those in the analysis above. However, theoretical and numerical support for the stability of $Q_2$--$P_0$ discretizations of $\ddirector$ and $\dlambda$ was given in \cite{Emerson1}. Furthermore, the lemmas proved above demonstrate that the discretization space for $\dphi$, in both the electric and flexoelectric models, may be arbitrarily chosen without regard for stability.
\vspace{.3cm}
\begin{algorithm}[H] \label{algo}
\SetAlgoLined
~\\
0. Initialize $(\director_0, \phi_0, \lambda_0)$ on coarse grid.
~\\
\While{Refinement limit not reached}
{
	\While{First-order optimality conformance threshold not satisfied}
	{
		1. Set up discrete linear system \eqref{ElectricPotentialNewtonSystem} on current grid, $H$. ~\\
		2. Solve for $\ddirector_{H}$, $\dphi_{H}$, and $\dlambda_{H}$. ~\\
		3. Compute $\director_{k+1}$, $\phi_{k+1}$, and $\lambda_{k+1}$ as in \eqref{corrections}. ~\\
	}
	4. Uniformly refine the grid. ~\\
	5. Interpolate $\director_{H} \to \director_{h}$, $\phi_{H} \to \phi_h$, and $\lambda_{H} \to \lambda_h$.
}
\caption{Newton's method minimization algorithm with NI}
\end{algorithm}
\vspace{.3cm}

\subsection{Coupled Multigrid with Vanka-type Relaxation}

Significant research into the development of efficient iterative solvers for block structures such as those arising in \eqref{decomphessian} and \eqref{flexodecomphessian} exists. Here, we discuss the implementation and results for a coupled multigrid method with Vanka-type relaxation. The performance and robustness of such methods have been studied in-depth for block linear systems pertaining to incompressible flows \cite{Larin1, Matthies1, John1}. Furthermore, these methods have been shown to achieve desirable convergence rates for systems with coupled saddle-point structures such as those in \eqref{decomphessian} and \eqref{flexodecomphessian} \cite{Benson1}. In this section, we write the general system to be solved as
\begin{equation*}
M 
\left [ \begin{array}{c}
\director \\
\phi \\
\lambda
\end{array} \right] =
\left [ \begin{array}{c c c}
A & B_1 & B_2 \\
B_1^T & -\tilde{C} & \vec{0} \\
B_2^T & \vec{0} & \vec{0}
\end{array} \right ]
\left [ \begin{array}{c}
\director \\
\phi \\
\lambda
\end{array} \right]  = 
\left [ \begin{array}{c}
f_{\director} \\
f_{\phi} \\
f_{\lambda}
\end{array} \right] ,
\end{equation*}
where $M$ represents a matrix arising for either the electric or flexoelectric models.

Due to the use of cell-centered, discontinuous finite elements for the Lagrange multiplier, the Vanka-type relaxation techniques herein, originally formulated in \cite{Vanka1} for finite-difference discretizations, are \emph{mesh-cell oriented}. Therefore, in the construction of the Vanka-type relaxation block associated with each Lagrange multiplier degree of freedom, all director and electric potential degrees of freedom associated with the same cell are considered. Let $\mathscr{N}_h$, $\mathscr{E}_h$, and $\mathscr{Q}_h$ denote the director, electric potential, and Lagrange multiplier degrees of freedom, respectively. Define $\mathscr{V}_{hj}$ to be the set of degrees of freedom associated with mesh cell $j$. Let $M_j$ be the block of matrix $M$ formed by extracting the rows and columns of $M$ corresponding to the degrees of freedom in $\mathscr{V}_{hj}$. Hence,
\begin{equation} \label{vankablock}
M_j = \left [
\begin{array}{c c c}
A_j & B_{1,j} & B_{2,j} \\
B_{1,j}^T & -\tilde{C}_{j} & \vec{0} \\
B_{2,j}^T & \vec{0} & \vec{0}
\end{array}
\right],
\end{equation}
with dimension $\vert \mathscr{V}_{hj} \vert \times \vert \mathscr{V}_{hj} \vert$. Solution values for degrees of freedom in $\mathscr{V}_{hj}$ are updated as
\begin{equation*}
\left [ \begin{array}{c}
\director_{i+1} \\
\phi_{i+1} \\
\lambda_{i+1}
\end{array} \right ]_j  = 
\left [ \begin{array}{c}
\director_{i} \\
\phi_{i} \\
\lambda_{i}
\end{array} \right ]_j 
+ \zeta M_j^{-1} 
\left ( \left [ \begin{array}{c}
f_{\director} \\
f_{\phi} \\
f_{\lambda}
\end{array} \right ]
- M
\left [ \begin{array}{c}
\director_{i} \\
\phi_{i} \\
\lambda_{i}
\end{array} \right ] 
 \right )_j,
\end{equation*}
where the subscript $j$ restricts the vectors to the appropriate rows. Thus, a single relaxation step consists of a loop over all mesh elements in the domain. 

Within the underlying multigrid method, we use standard finite-element interpolation operators and Galerkin coarsening. For additional details on the numerical implementation of the multigrid method and associated relaxation schemes, see \cite{Benson1}.

The relaxation and convergence properties of element-wise Vanka-type relaxation techniques have been studied analytically for the Poisson, Stokes, and Navier-Stokes equations in \cite{Molenaar1, Schoberl1, Manservisi1, MacLachlan1, Sivaloganathan1}. Moreover, numerical experiments have shown good performance for electrically coupled systems with similar structure to those considered here \cite{Benson1}. An ``economy" Vanka-type relaxation approach, as described in \cite{Benson1}, is also quite effective for these problems but does not prove to be as efficient as the full Vanka relaxation described above. Therefore, only the full Vanka-type relaxation scheme is considered below.

In the following section, the performance of the multigrid methods using the full Vanka-type relaxation technique is compared against that of using the UMFPACK LU decomposition \cite{TADavis1,TADavis2, TADavis3, TADavis4}, linked through the deal.II library, as an exact solver. Additionally, we consider a number of problems involving both the electric and flexoelectric models examined above and apply the coupled multigrid algorithm.

\section{Numerical Results} \label{numresults}

The general test problem in this section considers a classical domain with two parallel substrates placed at distance $d=1$ apart. The substrates run parallel to the $xz$-plane and perpendicular to the $y$-axis. It is assumed that this domain represents a uniform slab in the $xy$-plane. That is, $\director$ may have a non-zero $z$ component but $\pd{\director}{z} = \vec{0}$. Hence, we consider the 2-D domain $\Omega = \{ (x,y) \text{ } \vert \text{ } 0 \leq x,y \leq 1 \}$. The problem assumes periodic boundary conditions at the edges $x=0$ and $x=1$. Dirichlet boundary conditions are enforced on the $y$-boundaries. As discussed above, the simplification outlined in \eqref{functional3} is relevant for this domain and boundary conditions.

\subsection{Full Vanka Relaxation Studies}

In this section, we present results of relaxation parameter and solve time studies comparing the performance of the multigrid method using full Vanka-type relaxation against that of the UMFPACK LU decomposition exact solver. The studies were performed on a flexoelectric problem with relevant constants detailed in Table \ref{vankarelevantconstants}. Letting $r = 0.25$ and $s = 0.95$, the boundary conditions were
\begin{align}
n_1 &= 0, \label{nanopatterning1} \\
n_2 &= \cos\big(r(\pi + 2 \tan^{-1}(X_m) -2 \tan^{-1}(X_p))\big), \\
n_3 &= \sin\big(r(\pi + 2 \tan^{-1}(X_m) -2 \tan^{-1}(X_p))\big), \label{nanopatterning3}
\end{align}
where $X_m=\frac{-s\sin(2\pi(x+r))}{-s\cos(2\pi(x+r))-1}$ and $X_p = \frac{-s\sin(2\pi(x+r))}{-s\cos(2\pi(x+r))+1}$. Such boundary conditions are meant to simulate nano-patterned surfaces important in current research \cite{Atherton1, Atherton2}; see the substrate boundaries in Figure \ref{VankaParameterTiming}\subref{fig:right}.
Even in the absence of electric fields, such patterned surfaces result in complicated director configurations throughout the interior of $\Omega$. 
\begin{table}[h!]
\centering
{\small
\begin{tabular}{|c|c|c|c|c|c|}
\hline
Elastic Constants & $K_1 = 1$ & $K_2=4$ & $K_3= 1$ & $\kappa=4$ & $\epsilon_0 = 1.42809$ \\
\hline
Electric Constants & $\epsilon_{\parallel}=7$ & $\epsilon_{\perp}=7$ & $\epsilon_a = 0$ & $e_s = 0.5$ & $e_b = 0.5$ \\
\hline
\end{tabular}
}
\caption{\small{Relevant liquid crystal constants for Vanka-type relaxation studies.}}
\label{vankarelevantconstants}
\end{table}

The first set of studies focus on determining the optimal Vanka relaxation parameter $\zeta$. For these numerical experiments, the multigrid convergence tolerance, which is based on the ratio of the current solution's residual to that of the initial guess, is $10^{-6}$ for each grid level and Newton step. The relaxation parameter for the full Vanka approach is varied from $\zeta = 0.1$ to $\zeta = 1.1$ in increments of $0.05$. The corresponding average multigrid iteration counts for a $512 \times 512$ grid and a selection of $\zeta$ values is displayed in Figure \ref{VankaParameterTiming}\subref{fig:left} alongside the final computed solution in Figure \ref{VankaParameterTiming}\subref{fig:right}. 
\begin{figure}[h!]
\centering
\begin{subfigure}[b]{.49 \textwidth}
\raggedleft
  \includegraphics[scale=.35]{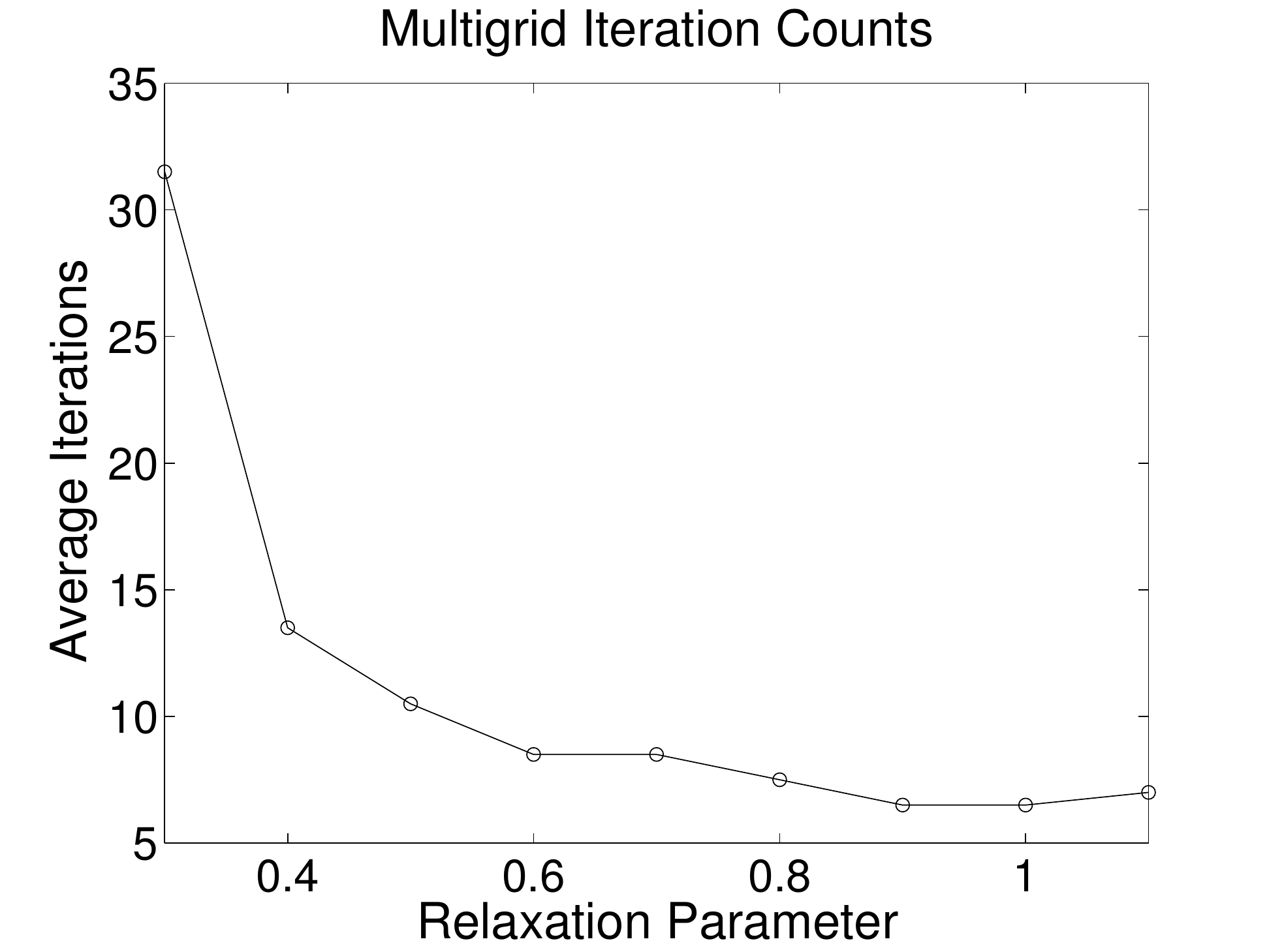}
  \caption{}
  \label{fig:left}
\end{subfigure}
\begin{subfigure}[b]{.49 \textwidth}
\raggedright
  \includegraphics[scale=.30]{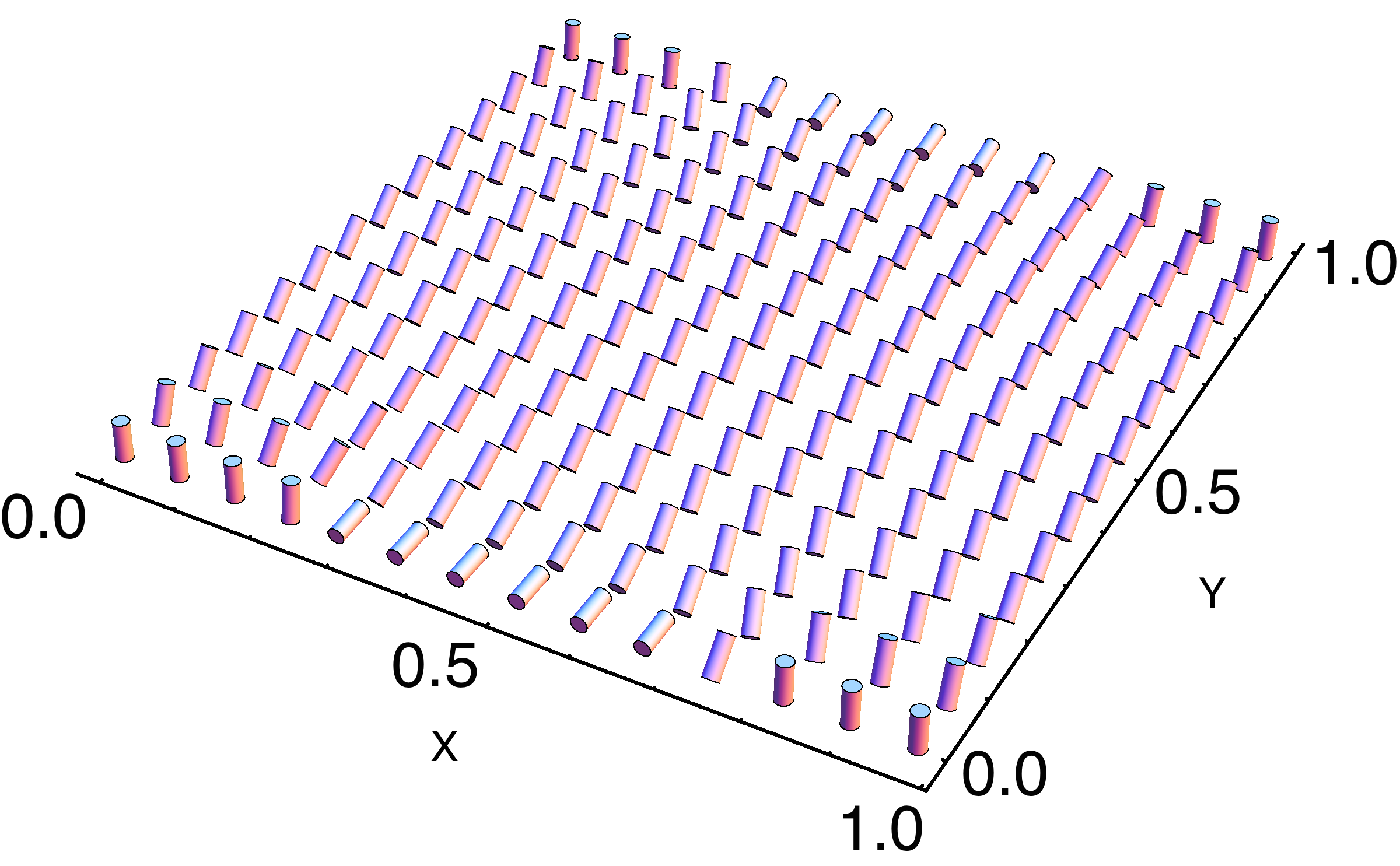}
  \caption{}
  \label{fig:right}
\end{subfigure}
\caption{\small{(\subref{fig:left}) The average number of multigrid iterations for varying $\zeta$ relaxation parameters on a $512 \times 512$ grid. (\subref{fig:right}) The final computed solution for the test problem on $512 \times 512$ mesh (restricted for visualization).}}
\label{VankaParameterTiming}
\end{figure}
For the figure, relaxation parameters smaller than $0.3$ are not included, as they resulted in iteration counts of over $100$ before the multigrid residual tolerance was satisfied. The studies indicate that a relaxation parameter of $\zeta = 1.00$ is optimal for convergence.

The second set of numerical experiments compares the system solve times for the Vanka-type solver against the performance of the UMFPACK LU decomposition solver utilized by deal.II. The experiments compare the linear solvers, on the above problem, with full nested iteration beginning on an $8 \times 8$ grid uniformly refining to a $512 \times 512$ mesh. For both solvers considered, we report the total time to solution, including both the setup and solve phases of the algorithms, but neglect some overhead associated with converting data formats and interfacing libraries. The optimal relaxation parameter, $\zeta = 1.00$, was used for the full Vanka-type relaxation technique. We consider multigrid methods using standard residual-based stopping tolerances, fixed on all grids, of reduction in the linear residual by factors of $10^{-8}$, $10^{-6}$, and $10^{-4}$. 

\begin{table}[h!]
\centering
{\small
\begin{tabular}{|c|c|c|c|c|c|c|c|}
\hline
 Solver\textbackslash Grid& $8\times8$ & $16 \times 16$ & $32 \times 32$ & $64 \times 64$ & $128 \times 128$ & $256 \times 256$ & $512 \times 512$ \\
\hline
LU & $0.02$ & $0.11$ & $0.57$ & $2.67$ & $12.03$ & $55.78$ & $275.86$ \\
\hline
Full 1e-8 & $0.05$ & $ 0.23$& $ 1.17$& $ 4.87$& $ 20.18$& $ 82.77$& $337.84$ \\
\hline
Full 1e-6 & $0.05$ & $0.20$& $ 0.91$& $ 3.78$& $ 16.72$& $ 66.38$& $ 276.91$ \\
\hline
Full 1e-4 & $0.04$ & $ 0.17$& $ 0.74$& $ 3.06$& $ 13.13 $& $ \mathbf{54.39}$& $ \mathbf{214.14}$ \\
\hline
\end{tabular}
}
\caption{\small{Comparison of average time to solution (in seconds) with LU decomposition (LU) and full Vanka relaxation (Full) for varying grid sizes. Numbers following the relaxation type indicate the multigrid residual tolerance. Bold face numbers indicate improved time to solution compared with the LU decomposition solver.}}
\label{VankaStaticTiming}
\end{table}

Table \ref{VankaStaticTiming} displays the average time to solution for the linear systems arising on successive grids. In the table, the multigrid solve timing is scaling nearly perfectly with grid size, while the LU decomposition solve times are growing at a faster rate. For the present timings, the LU decomposition solver is approximately scaling with a factor of $5$, and has an expected asymptotic scaling factor of $8$. The table also displays a clear confluence of the solve time for LU decomposition and the Vanka-type solver. For a multigrid residual tolerance of $10^{-4}$, the time to solution becomes nearly equal to that of the LU decomposition solver as early as the $128 \times 128$ grid. Moreover, though the applied Vanka-type relaxation method is an approximate linear solver, the number of overall Newton steps does not increase for any of the experiments compared to the direct solver. Therefore, the method is robust with respect to adjustments in the multigrid tolerance.

The results of these studies suggest that the full Vanka-type relaxation method discussed above is an effective, efficient, and scalable iterative solver applicable to the coupled saddle-point linear systems arising in the discretization of the electric and flexoelectric models. Furthermore, the relaxation technique exhibits notable performance for a range of multigrid residual tolerances and relaxation parameters. In the numerical simulations to follow, full Vanka relaxation is applied with a relaxation parameter of $1.00$ and a multigrid residual tolerance of $10^{-6}$ for assured accuracy. 

\subsection{Simple Electric Freedericksz Transition Results}

The first liquid crystal numerical experiment considers simple director boundary conditions, such that $\director$, along both of the substrates, lies uniformly parallel to the $x$-axis. The boundary conditions for the electric potential, $\phi$, are such that $\phi=0$ on the lower substrate at $y=0$ and $\phi=1$ at $y=1$. The relevant constants for the problem are detailed in Table \ref{relevantConstants}. Since the electric anisotropy constant, $\epsilon_a$, is positive, the expected behavior for the liquid crystal configuration is a Freedericksz transition \cite{Freedericksz1, Zocher1} so long as the applied field is strong enough to overcome the inherent elastic effects of the system. That is, for an applied voltage above a critical threshold, known as a Freedericksz threshold \cite{Stewart1}, the liquid crystal configuration will depart from uniform alignment parallel to the $x$-axis and instead tilt in the direction of the applied field. The problem considered has an analytical solution \cite[pg. 92-93]{Stewart1} demonstrating this behavior. The critical voltage is given by $V_c = \pi \sqrt{\frac{K_1}{\epsilon_0 \epsilon_a}}$. For the constants detailed in Table \ref{relevantConstants}, this implies a Freedericksz threshold of $0.7752$. Thus, the anticipated solution should demonstrate a true Freedericksz transition away from uniform free elastic alignment. Indeed, the final computed solution in Figure \ref{FreederickszTrans}, displayed alongside the initial guess for the algorithm, displays the expected transition.

\begin{table}[h!]
\centering
{\small
\begin{tabular}{|c|c|c|c|c|}
\hline
Elastic Constants & $K_1 = 1$ & $K_2=0.62903$ & $K_3= 1.32258$ & $\kappa=0.475608$\\
\hline
Electric Constants & $\epsilon_0 = 1.42809$ & $\epsilon_{\parallel}=18.5$ & $\epsilon_{\perp}=7$ & $\epsilon_a = 11.5$ \\
\hline
\end{tabular}
}
\caption{\small{Relevant liquid crystal constants for Freedericksz transition problem.}}
\label{relevantConstants}
\end{table}

\begin{figure}[h!]
\centering
\begin{subfigure}{.49 \textwidth}
\raggedleft
  \includegraphics[scale=.30]{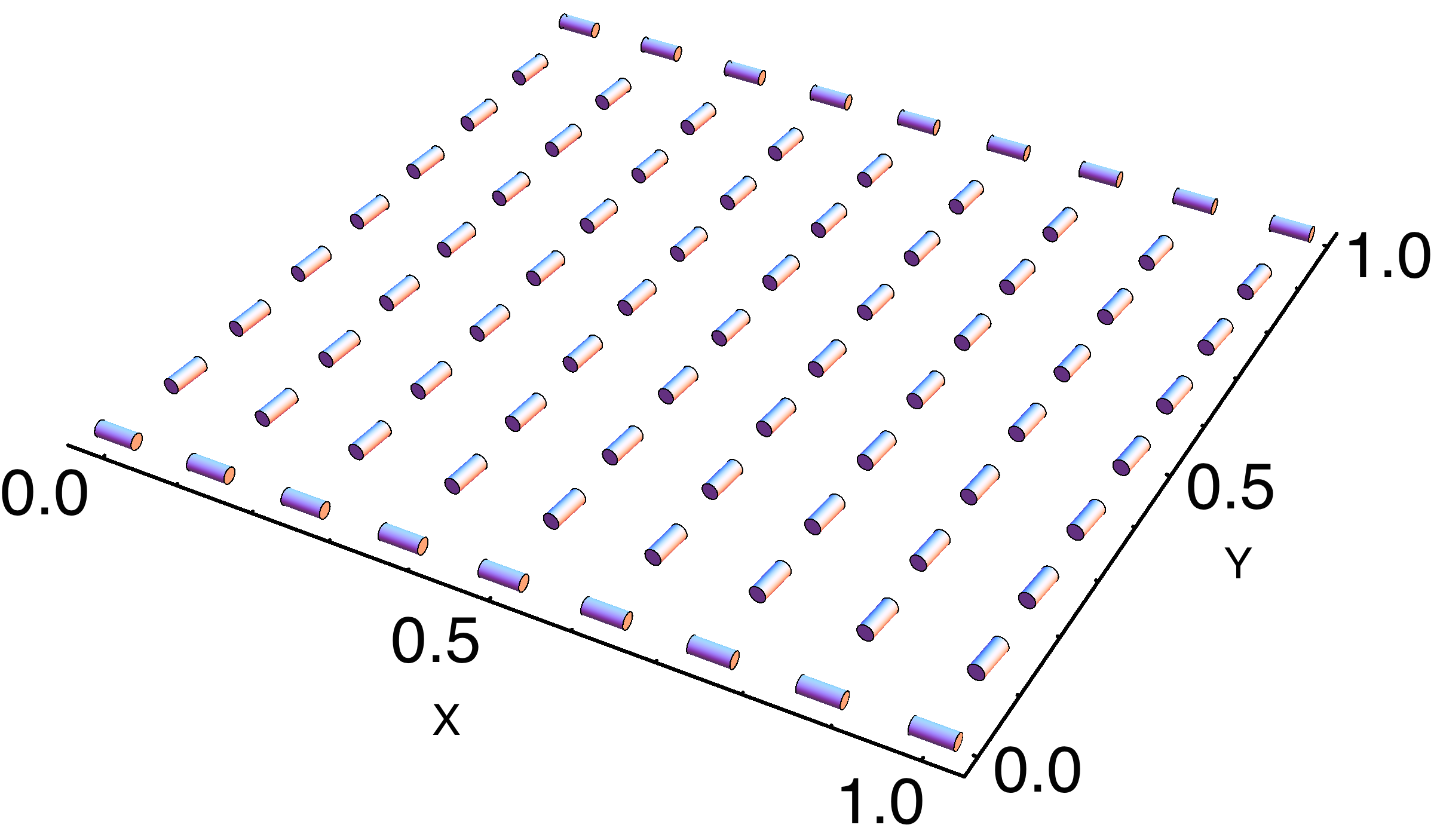}
    \caption{}
  \label{fig:left2}
\end{subfigure}
\begin{subfigure}{.49 \textwidth}
\raggedright
  \includegraphics[scale=.30]{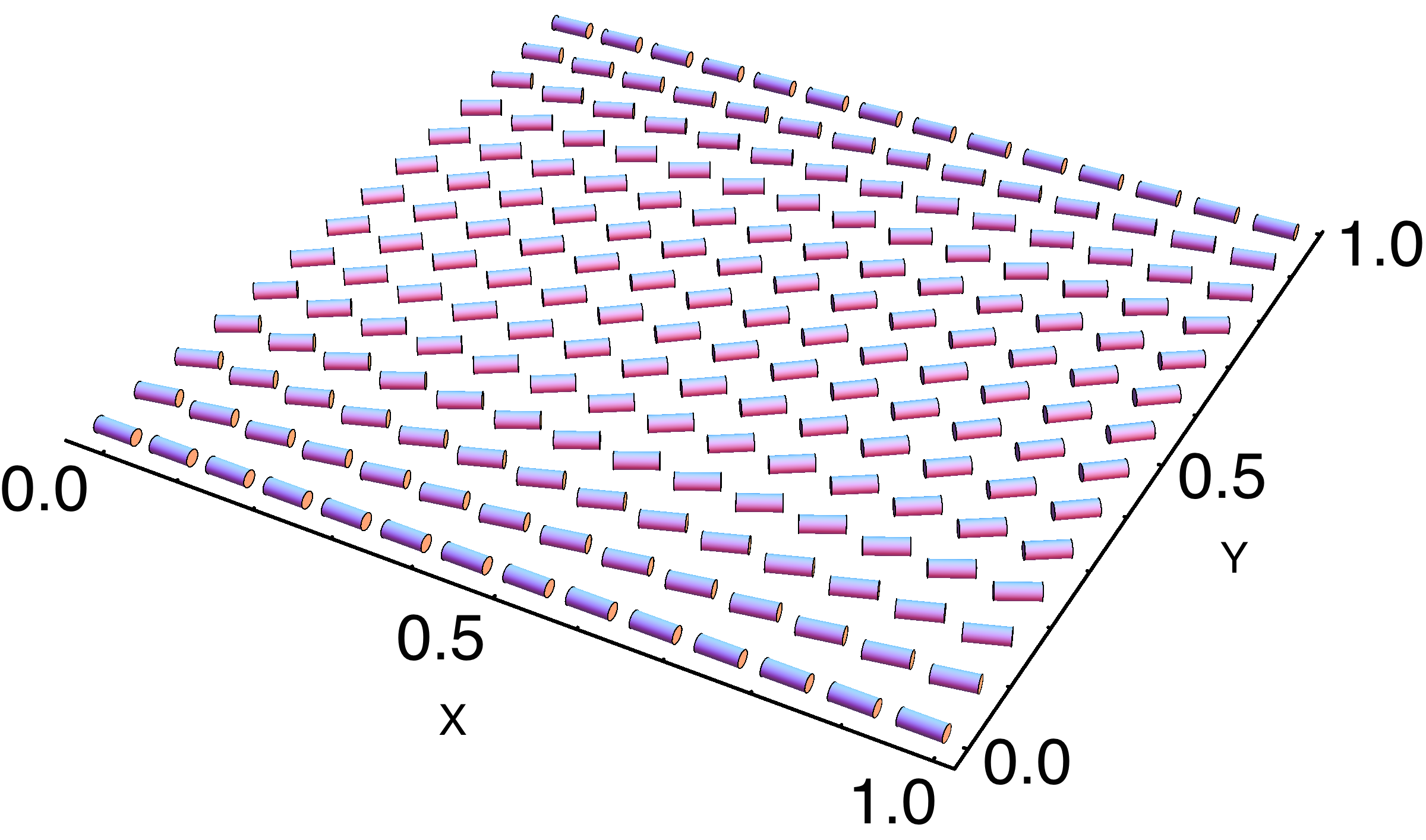}
    \caption{}
  \label{fig:right2}
\end{subfigure}
\caption{\small{(\subref{fig:left2}) Initial guess on $8 \times 8$ mesh with initial free energy of $26.767$ and (\subref{fig:right2}) resolved solution on $512 \times 512$ mesh (restricted for visualization) with final free energy of -5.330 for Freedericksz transition.}}
\label{FreederickszTrans}
\end{figure}

The problem is solved on a $8 \times 8$ coarse grid with six successive uniform refinements resulting in a $512 \times 512$ fine grid. The minimized functional energy is $\mathcal{F}_2 = -5.330$, compared to the initial guess energy of $26.767$. Figure \ref{ErrorandNIcounts}\subref{fig:left3} details the number of Newton iterations necessary to reduce the (nonlinear) residual below the given tolerance, $10^{-3}$, on each grid. Note that a sizable majority of the Newton iteration computations are isolated to the coarsest grids, with the finest grids requiring only one Newton iteration to reach the tolerance limit. 
Without the use of nested iteration, the algorithm requires $53$ Newton steps on the finest grid, alone, to reach a similar error measure. The nested-iteration-Newton-multigrid method achieves an accurate solution in $10.5$ minutes, compared to a total run time of over $5$ hours for standard Newton-multigrid. 
This corresponds to a speed up factor of $29.6$ or a work requirement for the nested iterations equivalent to $1.79$ times that of assembling and solving a single linearization step on the finest grid.
\begin{figure}[h!]
\centering
\begin{subfigure}{.49 \textwidth}
\raggedleft
  \includegraphics[scale=.35]{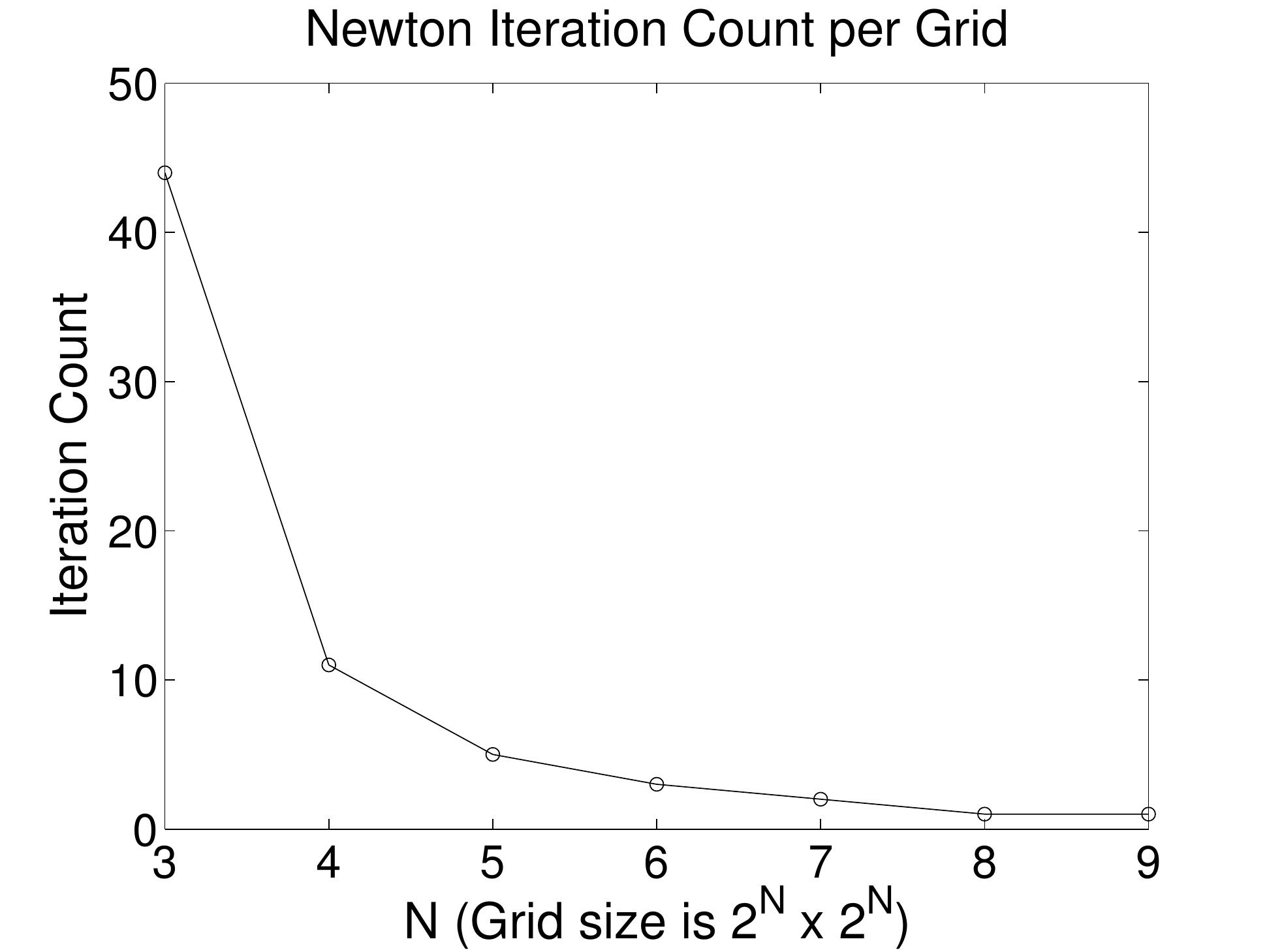}
      \caption{}
  \label{fig:left3}
\end{subfigure}
\begin{subfigure}{.49 \textwidth}
\raggedright
        \includegraphics[scale=.35]{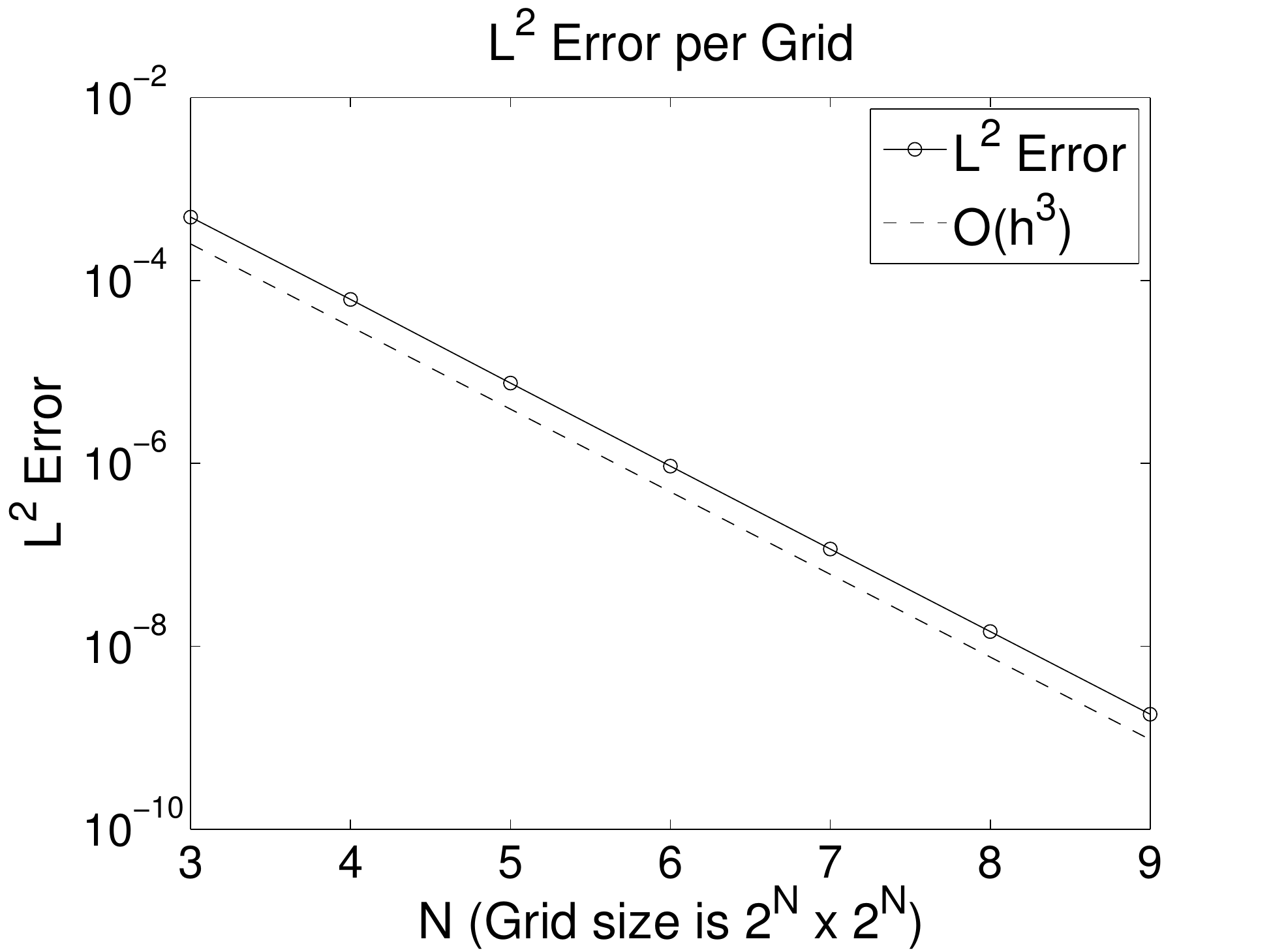}
      \caption{}
  \label{fig:right3}
\end{subfigure}
\caption{\small{(\subref{fig:left3}) Newton iterations and (\subref{fig:right3}) $L^2$-error per grid for the Freedericksz transition.}}
\label{ErrorandNIcounts}
\end{figure}

Also detailed in Figure \ref{ErrorandNIcounts}\subref{fig:right3} is the reduction in overall $L^2$-error comparing the analytical solution to the resolved solution on each grid. Note that the error is approximately reduced by a full order of magnitude on each successive grid, corresponding to approximately $O(h^3)$ reductions in overall error. Moreover, for the finer grids, a single Newton step was sufficient to achieve such a reduction.

\subsection{Electric Field with Patterned Boundary Conditions Results}

In the second liquid crystal run, the nano-patterned boundary conditions described by \eqref{nanopatterning1} - \eqref{nanopatterning3} are applied. The same constants outlined in Table \ref{relevantConstants} are also used for this problem. However, a stronger voltage such that $\phi=2$ on the substrate at $y=1$ is applied. Along the other substrate, $\phi$ remains equal to $0$. The final solution, as well as the initial guess, are displayed in Figure \ref{NanoBC}. For this problem, the grid progression again begins on an $8 \times 8$ grid ascending uniformly to a $512 \times 512$ fine grid. The minimized functional energy is $\mathcal{F}_2=-41.960$, compared to the initial guess energy of $-31.141$.
\begin{figure}[h!]
\centering
\begin{subfigure}{.49 \textwidth}
\raggedleft
  \includegraphics[scale=.30]{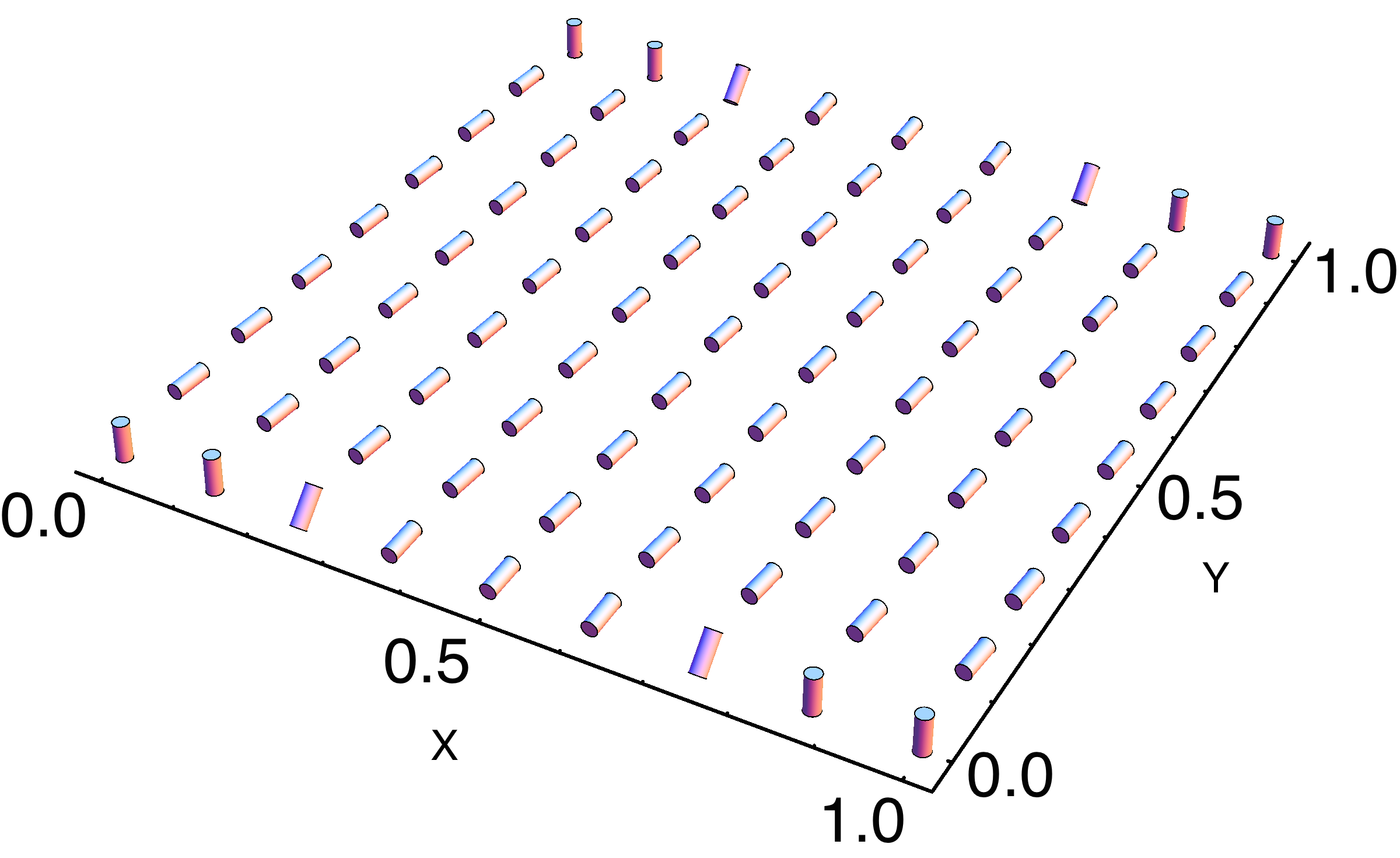}
        \caption{}
  \label{fig:left4}
\end{subfigure}
\begin{subfigure}{.49 \textwidth}
\raggedright
  \includegraphics[scale=.30]{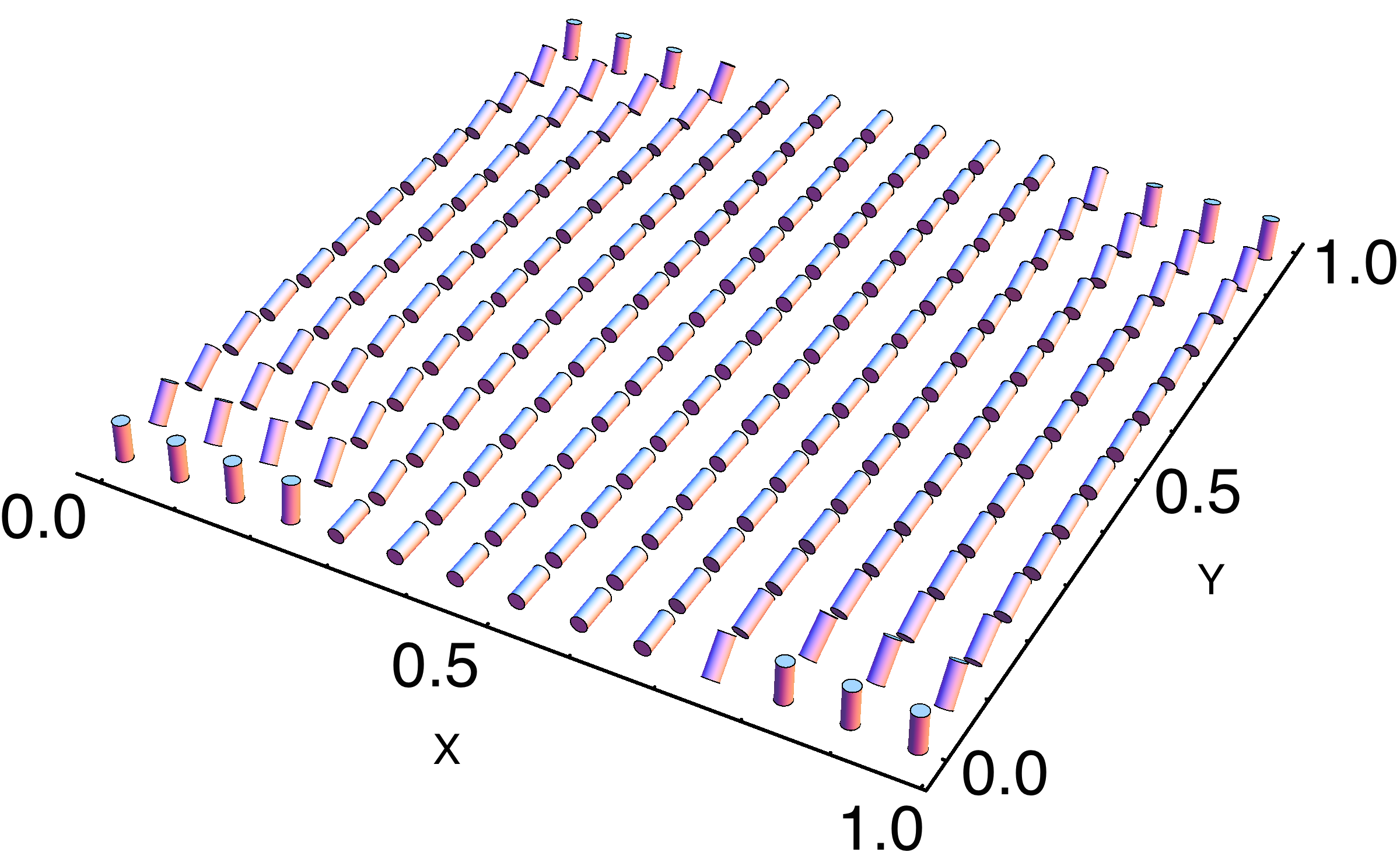}
        \caption{}
  \label{fig:right4}
\end{subfigure}
\caption{\small{(\subref{fig:left4}) Initial guess on $8 \times 8$ mesh with initial free energy of $-31.141$ and (\subref{fig:right4}) resolved solution on $512 \times 512$ mesh (restricted for visualization) with final free energy of -41.960 for nano-patterned boundary.}}
\label{NanoBC}
\end{figure}

\begin{table}[h!]
\centering
{\small
\begin{tabular}{|c|c|c|c|c|c|}
\hline
Grid Dim. &Newton Iter.&Init. Res.&Final Res.& Deviation in $\ltwonorm{\director}^2$ & Final Energy\\
\hline
$8 \times 8$ & 44 & 12.27e-00 & 6.79e-04 & -1.11e-01, 6.11e-02 & -42.701\\
\hline
$16 \times 16$ & 16 & 2.01e-00 & 5.74e-04 & -7.64e-02, 4.24e-02 &-42.170\\
\hline
$32 \times 32$ & 9 & 9.91e-01 & 2.60e-04 & -4.60e-02, 2.92e-02 &-41.963\\
\hline
$64 \times 64$ & 5 & 5.52e-01 & 1.76e-04 & -1.80e-02, 1.31e-02 &-41.950\\
\hline
$128 \times 128$ & 2 & 2.36e-01 & 3.13e-09 & -3.63e-03, 2.89e-03 &-41.960\\
\hline
$256 \times 256$ & 2 & 7.26e-02 & 1.65e-10 & -4.92e-04, 3.62e-04 &-41.960\\
\hline
$512 \times 512$ & 2 & 1.87e-02 & 6.10e-12 & -7.37e-05, 6.39e-05 &-41.960\\
\hline
\end{tabular}
}
\caption{\small{Grid and solution progression for electric problem and a nano-patterned boundary with initial and final residuals for the first-order optimality conditions, minimum and maximum director deviations from unit length at the quadrature nodes, and final functional energy on each grid.}}
\label{gridprogNanoBC}
\end{table}
\vspace{-.2in}
In Table \ref{gridprogNanoBC}, the number of Newton iterations per grid is detailed as well as the conformance of the solution to the first-order optimality conditions after the first and final Newton steps, respectively, on each grid. As with the previous example, much of the computational work is relegated to the coarsest grids. Here, the total work required is approximately $5.10$ times that of assembling and solving a single linearization step on the finest grid. In contrast, without nested iteration, the algorithm requires $52$ Newton steps on the $512 \times 512$ fine grid, alone, to satisfy the tolerance limit. While the nested-iteration-Newton-multigrid method achieves convergence in $30.4$ minutes, the standard Newton-multigrid total run time is over $5.3$ hours. Also shown in Table \ref{gridprogNanoBC}, the minimum and maximum director deviations from unit length at the quadrature nodes is descending towards zero.

Due to the sizable applied electric field, and the elastic influence of the central boundary condition pattern aligned with the electric field, the expected configuration is a quick transition from the boundary conditions to uniform alignment with the field. That is, the strength of the Freedericksz transition on the interior of $\Omega$ is augmented by the presence of this type of patterned boundary condition. This behavior is accurately resolved in the computed solution.

\subsection{Flexoelectric Phenomena}

As discussed above, internally generated electric fields due to flexoelectricity are an important physical aspect of liquid crystal configurations. This polarization due to curvature can significantly affect stable liquid crystal configurations in the presence of certain boundary conditions, such as patterned surfaces that cause large distortions in the nematic. These may also cause physical phenomenon such as bistability \cite{Atherton1, Atherton2, Davidson1} that are important for display applications. 

The following numerical results utilize similar boundary conditions to those in \eqref{nanopatterning1}-\eqref{nanopatterning3} with an extra parameter\footnote{Note that, here, $\varphi$ is utilized for the azimuthal angle whereas in \cite{Atherton1}, $\phi$ was used.}, $\varphi$, which has the effect of varying the imposed azimuthal director angle along the $x$-axis of the outer, vertically-aligned strips on the boundary,
\begin{align*}
n_1 &= \sin(\varphi)\sin\big(r(\pi + 2 \tan^{-1}(X_m) -2 \tan^{-1}(X_p))\big),\\
n_2 &= \cos\big(r(\pi + 2 \tan^{-1}(X_m) -2 \tan^{-1}(X_p))\big), \\
n_3 &= \cos(\varphi)\sin\big(r(\pi + 2 \tan^{-1}(X_m) -2 \tan^{-1}(X_p))\big).
\end{align*} 
The NI progression from $8 \times 8$ grids to $512 \times 512$ grids persists for each of the simulations. Due to the complexity of the flexoelectric systems, the nonlinear residual stopping tolerance is decreased to $10^{-5}$.

\begin{figure}[h!]
\centering
\begin{subfigure}{\textwidth}
\centering
  \includegraphics[scale=.35]{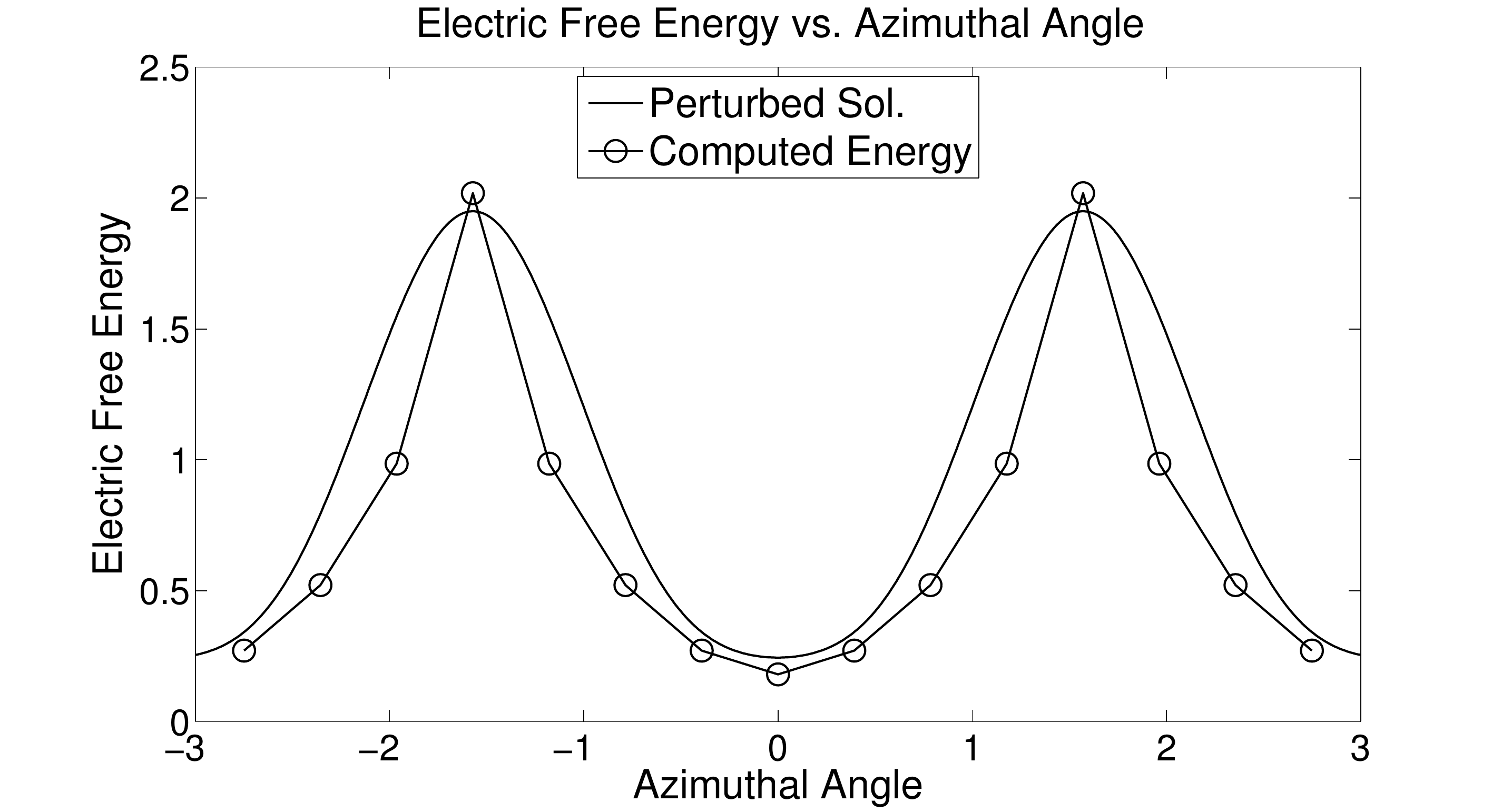}
\end{subfigure}
\caption{\small{The computed final free energy of the perturbative solution with $K_1 = K_2 = K_3 = 1$, $e_s = 5$, and $e_b =-5$ for varying $\varphi$ values.  A perturbation solution similar to that given in \cite{Atherton1} is overlaid}}
\label{FlexoEnergiesAnalyticalSolsa}
\begin{subfigure}{\textwidth}
\begin{center}
  \includegraphics[scale=.35]{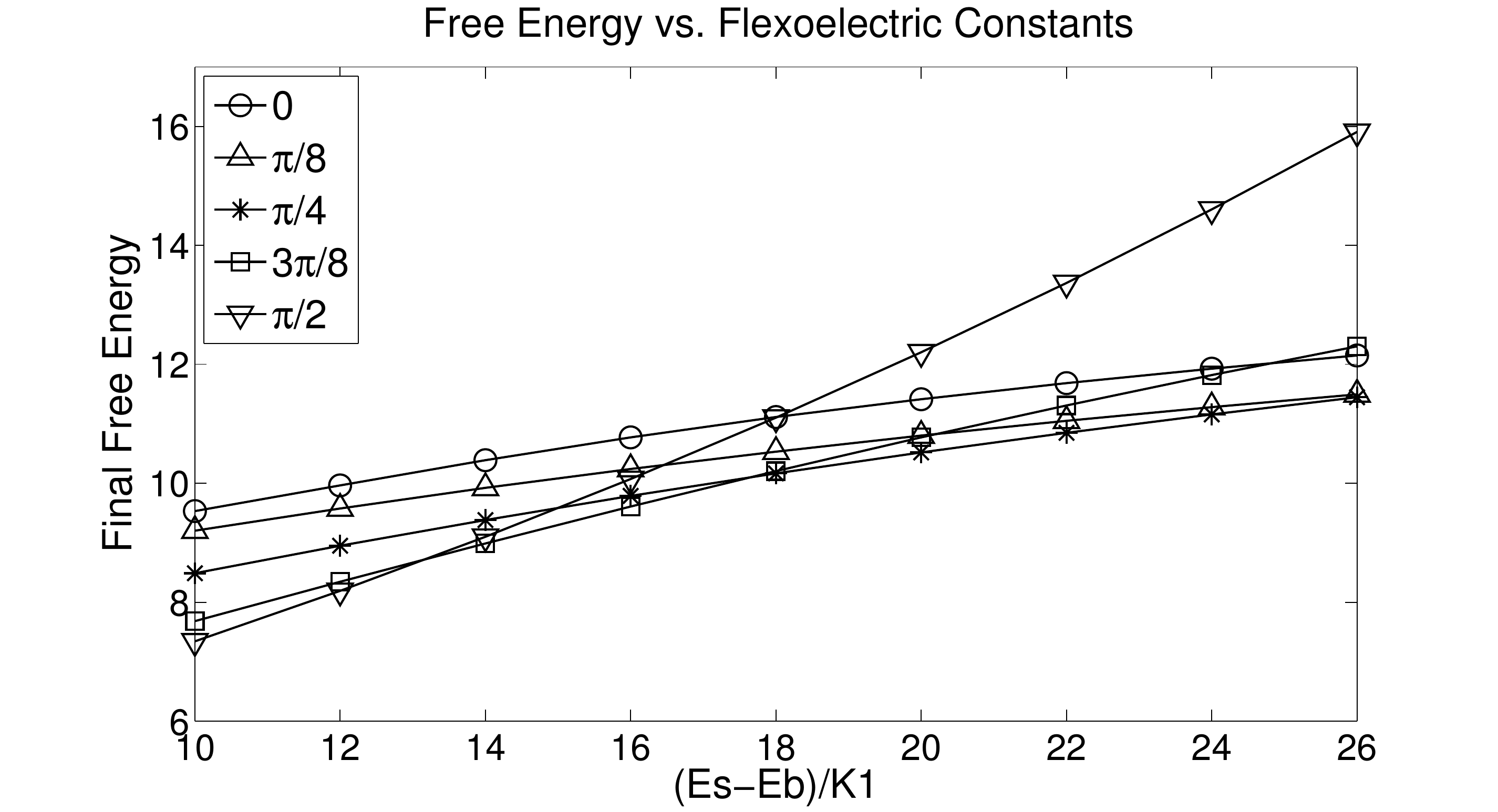}
  \end{center}
\end{subfigure}
\caption{\small{Final flexoelectric energies with nano-patterned boundary conditions for varying Rudquist constants $e_s$ and $e_b$. Each line corresponds to a different $\varphi$ value.}}
\label{FlexoEnergiesAnalyticalSolsb}
\end{figure}

In the first experiment, we isolate the influence of flexoelectricity on the configuration by removing elastic anisotropy, setting $K_1=K_2=K_3=1$, and using a small dielectric anisotropy $\epsilon_{\parallel} = 7$ and $\epsilon_{\perp} = 6.9$. For both experiments, as above, $\epsilon_0 = 1.42809$. The computed free energy as a function of the azimuthal angle $\varphi$ is shown in Figure \ref{FlexoEnergiesAnalyticalSolsa}, revealing that $\varphi=0$ and $\varphi=\pi$ are the minima, corresponding to alignment along the length of the stripes. Hence, flexoelectricity serves as an aligning effect in the presence of the patterned surface.  Also displayed in the figure is the free energy of a perturbation solution similar to the one derived in \cite{Atherton1} (note, a different unit convention and sign error exists in \cite{Atherton1}). There, the perturbation solution is valid for a single semi-infinite planar-vertical junction.  In the numerical computation, the director profile for the striped cell consists of four junctions per unit cell.  Thus, we approximate the perturbation by adding the mirror image and doubling.  If the junctions are well separated from each other, the cell thickness is larger than the penetration depth of the nematic, and the length of the surface planar-vertical transition is very small, this is a valid approximation.  Even within this limitation, though, the computed energies trace the characteristics of the perturbation solution quite closely, verifying the alignment influence of flexoelectricity. Therefore, when considering internally induced electric fields in the presence of nano-patterned boundaries, the algorithm's computed free energies capture the qualitative prediction from the perturbation solution, but do so with a quantitative accuracy that is not readily matched by perturbation techniques.

For the second experiment, $\epsilon_{\parallel} = 7$ and $\epsilon_{\perp} = 7$. By including anisotropic elastic constants, it is possible to promote alignment perpendicular to the stripes, if $K_1,K_3 < K_2$, or parallel to the length of the stripes, if $K_1,K_3 > K_2$. We use $K_1=K_3=1$ and $K_2=4$ to select perpendicular alignment and simulate the configurations with $\varphi \in \{0, \frac{\pi}{8}, \frac{\pi}{4}, \frac{3\pi}{8}, \frac{\pi}{2} \}$ for varying values of the flexoelectric constants; the results are displayed in Figure \ref{FlexoEnergiesAnalyticalSolsb} . As can be seen, for $(e_b -e_s)/K_1 = 10$, the overall minimum of the free energy lies at an azimuthal angle $\varphi=\pi/2$ as expected. As the flexoelectric parameter is increased however, the configurations with different azimuthal angle increase at different rates; for example at a critical value of $(e_b -e_s)/K_1 \approx 17.5$, the solutions for $\varphi=0$ and $\varphi=\pi/2$ become degenerate.  Hence, as the strength of the flexoelectric effect is increased, the azimuthal angle corresponding to the ground state gradually rotates because flexoelectricity and elastic anisotropy favor opposing configurations. The phenomenon is important for applications because it may lead to multiple stable configurations in some regions of the parameter space, or a significant renormalization of the anchoring behavior for materials with large flexoelectric response. These phenomena allow engineers to control the ground states and, potentially, the switching response by adjusting the pattern. The above efficient numerical model would be a valuable tool in identifying the parameters that lead to the desired effect.

\section{Summary and Future Work} \label{conclusion}

We have discussed a constrained minimization approach to solving for liquid crystal equilibrium configurations in the presence of applied and internal electric fields. Such minimization is founded upon the electrically and flexoelectrically augmented Frank-Oseen models. Due to the nonlinearity of the continuum first-order optimality conditions, Newton linearizations were needed. The discrete Hessian arising in the finite-element discretization of these linearized systems was shown to be invertible, for both models, under certain assumptions on the bilinear forms. Using the finite-element spaces discussed in \cite{Emerson1}, these assumptions are satisfied. Additionally, an efficient iterative solvers utilizing a Vanka-type relaxation technique was implemented and shown to possess desirable solve timings and convergence properties for highly refined meshes.

Numerical results demonstrated the accuracy and efficiency of the algorithm in resolving both classical and complicated features induced by applied and internal electric fields. The method efficiently captured expected, complicated, physical phenomenon due to flexoelectric effects. In addition, the minimization approach overcomes some difficulties inherent to the liquid crystal equilibrium problem, such as the nonlinear unit length director constraint and effectively deals with heterogeneous Frank constants. The algorithm also productively utilizes nested iterations to reduce computational costs by isolating much of the computational work to the coarsest grids. Future work will include the study of effective adaptive refinement and linearization tolerance schemes. Further, investigation of line search and trust region algorithms and their performance will be undertaken.

\section*{Acknowledgments}
The authors would like to thank Professor Thomas Manteuffel for his useful contributions and suggestions.


\bibliographystyle{plain}	

\nocite{*}		

\bibliography{LiquidCrystalElectricField}		

\end{document}